\documentclass[reqno]{amsart}
\usepackage{amssymb}
\usepackage{amscd,enumerate,amsfonts,calc,amsmath,verbatim, euscript}
\usepackage{enumerate}
\usepackage{tikz}
\usepackage{xcolor}
\usepackage[hidelinks]{hyperref}
\hypersetup{colorlinks,breaklinks,
             urlcolor=blue,
             linkcolor=blue}
\usepackage[pagewise]{lineno}

\newtheorem{thm}{{\bf Theorem}}[section]
\newtheorem{lemma}[thm]{{\bf Lemma}}
\newtheorem{rmk}[thm]{{\bf Remark}}
\newtheorem{prop}[thm]{{\bf Proposition}}
\newtheorem{cor}[thm]{{\bf Corollary}}

\newcommand{\Z}{\mathbb Z}
\newcommand{\Ls}{\EuScript{L}}
\newcommand{\Ie}{\EuScript{I}}

\newcommand{\RM}{\mathrm{RM}}

\def\1{{1\hskip-0.25em{\rm 1}}}

\def\CC{{\ \rlap{\raise 0.4ex 
\hbox{$\scriptscriptstyle |$}}\hskip -0.2em C}}

\def\G{\mathcal{G}}
\def\M{\mathcal{M}}

\def\Z{\mathbb{Z}}
\def\prm{\mathsf{PRM}}

\def\Fq{{\mathbb F}_q}

\def\PP{{\mathbb P}}

\newcommand{\Mod}[1]{\ (\mathrm{mod}\ #1)}

\newcommand{\Ev}{\operatorname{Ev}}
\newcommand{\supp}{\mathrm{Supp}}

\newcommand{\ssupp}{\mathrm{supp}}

\newcommand{\stirling}[2]{\genfrac{[}{]}{0pt}{}{#1}{#2}}

\begin{document}

\title[Projective Reed-Muller Codes]{On the Minimum Distance,  Minimum Weight Codewords, and the Dimension of Projective Reed-Muller Codes}
\thanks{This article is an expanded version of the talk given by the first named author at the conference \href{https://conferences.cirm-math.fr/2804.html}{ALCOCRYPT}, held at CIRM, Luminy, France, in February 2023. The talk is \href{https://library.cirm-math.fr/Record.htm?idlist=1&record=19391979124911191519}{available} in the Audiovisual Mathematics Library of CIRM and on \href{https://www.youtube.com/watch?v=5edjTvqu2Rc}{YouTube}.}
	\author{Sudhir R. Ghorpade}
	\address{Department of Mathematics, 
		Indian Institute of Technology Bombay,\newline \indent
		Powai, Mumbai 400076, India}
\email{\href{mailto:srg@math.iitb.ac.in}{srg@math.iitb.ac.in}}
\thanks{Sudhir Ghorpade is partially supported by  the grant DST/INT/RUS/RSF/P-41/2021 from the Department of Science \& Technology, Govt. of India, and the IRCC award grant 12IRAWD009 from IIT Bombay.}
	\author{Rati Ludhani}
	\address{Department of Mathematics, 
		Indian Institute of Technology Bombay,\newline \indent
		Powai, Mumbai 400076, India}
\email{\href{mailto:ratiludhani@gmail.com}{ratiludhani@gmail.com}}
\thanks{Rati Ludhani is supported by Prime Minister's Research Fellowship PMRF-192002-256 at IIT Bombay.}

\begin{abstract}
We give an alternative proof of the formula for the minimum distance of a projective Reed-Muller code of an arbitrary order. It leads to a complete characterization of the minimum weight codewords of a projective Reed-Muller code. This is then used to determine the number of  minimum weight codewords of a projective Reed-Muller code. Various formulas for the dimension of a  projective Reed-Muller code, and their equivalences are also discussed. 
%
\end{abstract}

\maketitle
\section{Introduction}
Reed-Muller codes constitute a widely studied and fairly well-understood class of linear codes. These codes were introduced, in the binary case, by Muller \cite{M}, and further studied by Reed \cite{R} in 1954. (See \cite{R2} for an account of the early history.) Generalizations of Reed-Muller codes to the $q$-ary case were considered and extensively studied in the following decade by Kasami, Lin, and Peterson \cite{KLP} and by Delsarte, Goethals, and MacWilliams \cite{DGM}. As a result of this, and also subsequent work, notably by Berger and Charpin \cite{BC} and by Heijnen and Pellikaan \cite{HP}, we know a great deal about $q$-ary Reed-Muller codes of an arbitrary order.
For instance, the length, dimension, minimum distance, characterization and enumeration of minimum weight codewords, duals, automorphisms, and generalized Hamming weights, are completely known. 
For convenience of reader and ease of reference, a brief summary of several of the known results is given in \S \ref{subsec:RM}. 

Projective Reed-Muller (PRM) codes in the general $q$-ary case, but of order $d<q$, were introduced by Lachaud \cite{L} in 1986. (See Section~3 of \cite{GRRT} for an account of the history.) The length and dimension as well as bounds for the minimum distance of PRM codes of order $d<q$, and the exact minimum distance in the case $d=2$ and $q=2^e$, $e\ge 2$, were obtained by Lachaud \cite{L}.  Some of these ideas and results go back to Manin's course of lectures during 1981-82 at Moscow State University,  and are sketched in  Chapter I, Section 3,  Example 2 b) of \cite{VM}. Tsfasman conjectured a sharp bound for the number of $\Fq$-rational points of a projective hypersurface of a given degree $d\le q$. This bound was proved in the affirmative by Serre \cite{Se}, and it was then shown by Lachaud \cite{L2} how the minimum distance of PRM codes of order $d\le q$ can be obtained as a consequence of Serre's inequality. Independently, S\o rensen \cite{S} gave explicit formulas for the length, dimension, and the minimum distance of PRM codes not just of order $d\le q$, but of an arbitrary order. Further, he determined the duals of PRM codes as well. S\o rensen's work has had a considerable impact as is indicated by the number of citations to \cite{S}. Moreover, PRM codes continue to be of considerable current interest, and while the automorphism group is known (cf. \cite{B}), the question about the determination of their generalized Hamming weights is still open in general even though there has been significant progress lately (see, e.g., \cite{BDG2} and references therein for the current state of the art).

Recently, while revisiting  S\o rensen's proof for the minimum distance of PRM codes of an arbitrary order, we noticed that there was a gap in the argument. While we realized later that the gap can be fixed and the result is correct as stated, we were able to obtain an alternative proof of S\o rensen's result by suitably adapting the arguments of Serre \cite{Se}.
An additional dividend of this alternative proof is that it  leads to a characterization of the minimum weight codewords of a projective Reed-Muller code. Since such a characterization is not given in S\o rensen's paper \cite{S} and it did not seem to be readily available in the literature, we initially thought that it was new. However, we later found that Rolland \cite[Lem. $2.3$]{RR} 
has already given a characterization of the minimum weight codewords of a PRM code using different methods; see also Ballet and Rolland \cite[Thm. 8]{BR} for a detailed proof of Rolland's result.  Nonetheless, we believe that the alternative proof given here and the resulting characterization is of some interest. Another noteworthy feature of our proof is the use of the notion of projective reduction that was developed and systematically used in \cite{BDG1}, and this seems quite natural in the context of PRM codes.  We also give an explicit formula for the number of minimum weight codewords of a PRM code of an arbitrary order.  As far as we know, this appears to be new and somewhat nontrivial, even after knowing a characterization of the 
minimum weight codewords.  Finally,  we review various formulas for the dimension of a PRM code of an arbitrary order and make some remarks on their equivalence. 
We also include a direct proof of one of the dimension formulas using the notion of projective reduction.  
An appendix reproduces a hitherto unpublished proof of M.  Quercia about the equivalence of two classical formulas for the dimension of a PRM code. 

\section{Basics of Reed-Muller  and Projective Reed-Muller Codes}
\label{sec2}

Fix, throughout this paper, a prime power $q$ and a positive integer $m$. 
As usual, $\Z$ denotes the set of integers and  $\Fq$ the finite field with $q$ elements. Given any $b\in \Z$ and $a$ in an integral domain $\mathsf{R}$ containing $\Z$ as a subring, the corresponding binomial coefficient is defined~by
$$
{a \choose b} :=\begin{cases} 
\displaystyle{\frac{a(a-1) \cdots (a-b+1)}{b!}} & \text{ if } b \ge 0, \\[.75em]
0 & \text{ if } b < 0.
\end{cases}
$$
This is, in general, an element of the quotient field of $\mathsf{R}$. In case $\mathsf{R}=\Z$, it is in $\Z$.
%
%

\subsection{Reed-Muller Codes}
\label{subsec:RM}
Let $\nu\in \Z$ with  $\nu \ge 0$, and let $\Fq[X_1,\ldots,X_m]_{\le \nu}$ denote the space of all polynomials in $m$ variables with coefficients in $\Fq$ of (total) degree $\le \nu$. 
 Write $\Fq^m=\{{{P}}_1,\ldots,{{P}}_{q^m}\}$. Consider 
 the evaluation map 
\begin{eqnarray*}
&\text{ev}:\Fq[X_1,\ldots,X_m]_{\le \nu} \to \Fq^{q^m} \quad {\rm defined\; by\; }\quad f\mapsto c^{ }_f:=\left(f({{P}}_1),\ldots, f({{P}}_{q^m}) \right).
\end{eqnarray*}
%
Clearly $\text{ev}$ is a linear map and 
its image is a nondegenerate linear code of length $q^m$; this code 
is called the \emph{(generalized) Reed-Muller code} of order $\nu$, and it is denoted by $\RM_q(\nu,m)$.

It is easy to see that if  $\nu < q$, then the  map $\text{ev}$ is injective and so $\dim\ \mathrm{RM}_q(\nu,m) = {{m+\nu}\choose{\nu}}$; see, e.g., \cite[Lem. 2.1]{G2}.
Also, if  $\nu < q$,  then it follows from Ore's inequality (see, e.g., \cite[Cor. 2.7]{G2}) that the minimum distance of $\RM_q(\nu,m)$ is
  $(q - \nu) q^{m-1} $. 
On the other hand, if  $\nu\ge m(q-1)$, then it is not difficult to see (using, for instance, a slight modification of the argument in \cite[Rem. 2]{GR}) that the map $\mathrm{ev}$ is surjective and so $\mathrm{RM}_q(\nu,m) =\Fq^{q^m}$.  With this in view, we shall assume that $0\le \nu \le m(q-1)$. In this general case, the restriction of $\mathrm{ev}$ to the space of all \emph{reduced}\footnote{A polynomial $f\in \Fq[X_1,\ldots,X_m]$ is said to be \emph{reduced} if $\deg_{X_i}f \le q-1$ for each $i=1,\dots , m$.}  polynomials in $\Fq[X_1,\ldots,X_m]_{\le \nu}$ is injective (see, e.g., \cite[Lem. 2.1]{G2}); consequently, the  
dimension, say $\rho_q(\nu, m) $,  of $\RM_q(\nu,m)$ is given by the number of reduced monomials in $m$ variables of degree $\nu$, and it is explicitly given by either of the following formulas. 
\begin{equation}
\label{eq:dimRM}
\sum_{s=0}^{\nu}\sum_{i=0}^m (-1)^i \binom{m}{i} \binom{s-iq + m - 1}{s-iq} =
\sum_{i=0}^{m}(-1)^i \binom{m}{i}\binom{m+\nu -iq}{m}.
\end{equation}
One may refer to \cite[Thm.  5.4.1]{AssKey} for a proof of the first formula,  and to \cite[Prop.~5]{RT} or \cite[Lem. 1]{GR} for a direct proof of the second formula.

Kasami-Lin-Peterson \cite{KLP} showed that if $\nu\in \Z$ with $0\le \nu \le m(q-1)$  is written as 
${ \nu=t(q-1)+s}$ for unique $t,s\in \Z$ such that $t\ge 0$ and  $0\le s<q-1$,  then
the minimum distance of $\RM_q(\nu,m)$ is $ (q-s)q^{m-t-1}$. Furthermore, with 
 $t,s$ as above, Delsarte-Goethals-MacWilliams \cite{DGM} showed that $c\in  \mathrm{RM}_q(\nu,m)$ is a minimum weight codeword if and only if $c = \text{ev}(f)$, where
\begin{equation}
\label{eq:minwtrm}
f =  \omega_0\left(\prod_{i=1}^{t}(1-\ell_i^{q-1})\right)  \prod_{j=1}^{s} (\ell_{t+1}-\omega_j)
\end{equation}
for some linearly independent polynomials $\ell_1,\ldots,\ell_{t+1} \in \Fq[X_1, \dots , X_m]$ of degree~$1$, and some  
$\omega_0, \omega_1, \dots , \omega_s \in \Fq$ 
with $\omega_0\neq 0$ and $\omega_1, \dots , \omega_s$ distinct. For an alternative proof of this characterization, see \cite{Le}. Here it should be understood that when $s=0$, the second product in \eqref{eq:minwtrm} is empty (and hence equal to $1$) and in this case only the existence and linear independence of $t$ polynomials  $\ell_1,\ldots,\ell_{t} \in \Fq[X_1, \dots , X_m]$ of degree~$1$ is asserted in the characterization. Delsarte-Goethals-MacWilliams \cite{DGM} further observed that the number of minimum weight codewords of  $\RM_q(\nu,m)$ is given by
\begin{equation}
	\label{eq:NoMinWtRM}
	(q-1) q^t {{m}\brack{t}}_q M_s, \quad \text{where} \quad M_s:=\begin{cases}
	\displaystyle{	{q\choose s} { {m-t}\brack{1}}_q }& \text{ if } s >0, \\[.75em]
		1 & \text{ if } s = 0.
	\end{cases}
\end{equation}
Here, and hereafter, for any $a, b\in \Z$, we denote, as usual, by ${{a}\brack{b}}_q$ the Gaussian binomial coefficient given by 
$$
{{a}\brack{b}}_q : = \begin{cases} \displaystyle{ \frac{(q^a-1)(q^a-q) \cdots (q^a - q^{b-1})}{(q^b-1)(q^b-q) \cdots (q^b - q^{b-1})}} & \text{ if } 0\le b \le a, \\[.75em]
		0 & \text{ otherwise}.
	\end{cases}
$$
Delsarte-Goethals-MacWilliams \cite{DGM} also showed that the dual of 
$\mathrm{RM}_q(\nu,m)$ is given by $\mathrm{RM}_q(m(q-1) - \nu - 1,m)$, where 
by convention,  $\RM_q(-1,m)$ is the trivial code $\{\mathbf{0}\}$.
As indicated in the Introduction, the automorphism group as well as the generalized Hamming weights of $\RM_q(\nu, m)$ are known, and for these, we refer to the papers of 
 Berger and Charpin \cite{BC} and of Heijnen and Pellikaan \cite{HP} as well as Beelen and Datta \cite{BD}. 

\subsection{Projective Reed-Muller Codes}\label{subsec:PRM}
Fix, in the remainder of this paper, a positive integer $d$. Given any $k\in \Z$, let
$$
p^{ }_k:= \begin{cases} 1+q+q^2+\cdots+q^k &  \text{ if } k\ge 0, \\
 0 &  \text{ if } k<0. \end{cases}
 $$ 
 Each point of the $m$-dimensional projective space $\PP^m(\Fq)$ over $\Fq$ admits a unique representative in $\Fq^{m+1}$ in which the last
nonzero coordinate is $1$. Let $\mathsf{P}_1,\ldots, \mathsf{P}_{p^{ }_m}$ be an ordered listing of such representatives
in $\Fq^{m+1}$ of points of $\PP^m(\Fq)$. 
Let $\Fq[X_0,\ldots,X_m ]_d$ denote the space of homogeneous polynomials
in $\Fq[X_0, \dots ,X_m]$ of degree $d$ together with the zero polynomial.
Consider the evaluation map 
\begin{eqnarray}\label{eqprm:Ev}
&\Ev:\Fq[X_0,\ldots,X_m]_d \to \Fq^{{p^{}_m}} \quad {\rm defined\; by}\quad  F\mapsto c^{ }_F:=(F({\mathsf{P}}_1),\ldots, F({\mathsf{P}}_{p^{}_m})).
\end{eqnarray}
Note that $\Ev$ is a linear map 
and its image is a nondegenerate linear code of length $p^{ }_m$. This code is said to be 
the \emph{projective Reed-Muller code} of order $d$ and length $p^{ }_m$, and it is denoted by $\prm_q(d,m)$.

A convenient way to study the projective Reed-Muller code is to use the notion of projective reduction, which was introduced in \cite{BDG1} and employed in \cite{G2} to determine the vanishing ideal of $\PP^m(\Fq)$. Let us recall that the projective reduction $\overline{\mu}$ of a monomial $\mu \in \Fq[X_0,\ldots,X_m ]$  
of positive degree is defined as follows. Write 
$\mu =
X_0^{a_0}\cdots X_{\ell}^{a_{\ell}}$ for unique nonnegative integers $a_0, \dots , a_{\ell}$ with $a_{\ell}>0$ so that $X_{\ell}$ is the last variable in $\mu$. 
For $0 \le i \le \ell -1$, let $\overline{a}_{i}=a_{i}$ if $0\le a_{i} \le q-1$, whereas if $a_{i} \ge q$, then let 
$\overline{a}_{i} $ be the unique integer such that 
  $\overline{a}_{i} \equiv a_{i} \Mod{q-1}$ and 
$1\le \overline{a}_{i} \le q-1$. We then define  
$$
\overline{\mu}:=X_0^{\overline{a}_0}\cdots X_{\ell-1}^{\overline{a}_{\ell-1}}X_{\ell}^{a_{\ell}+\sum_{ i=0}^{\ell -1}(a_{i}-\overline{a}_{i})}.
$$
The monomial $\overline{\mu}$ is uniquely determined by $\mu$ and it has the same degree as  $\mu$. In case $\mu =1$,  we set $\overline{\mu}:=1$. 
Observe that $\overline{\mu}({\mathsf{P}}) = \mu({\mathsf{P}})$ for all ${\mathsf{P}}\in \Fq^{m+1}$. 
We say that a monomial $\mu$ is \emph{projectively reduced} if $\overline{\mu}=\mu$. 
These notions extend by $\Fq$-linearity to arbitrary polynomials in $\Fq[X_0,\ldots,X_m ]$. 
 Now let
	$$
		C_q(d,m):=\{f\in \Fq[X_0,\ldots,X_m]_d: f \text{ is projectively reduced}\}.
	$$
Note that $C_q(d,m)$ is a $\Fq$-linear subspace of the polynomial ring $\Fq[X_0,\ldots,X_m]$.

\begin{lemma}\label{lemp3}
The restriction ${\rm Ev}|_{C_q(d,m)}$ of the  evaluation map defined in \eqref{eqprm:Ev},   to $C_q(d,m)$, is injective.   Moreover, $\prm_q(d,m)={\rm Ev}(C_q(d,m))$. 
 
\end{lemma}
\begin{proof}
The first assertion follows from Lemma 3.1 of \cite{G2}, and 
the second assertion follows from Lemma 3.2 and Theorem 3.3 of \cite{G2}.
\end{proof}

\begin{cor}
\label{cor:lem1}
If  $F$ is a nonzero polynomial in $C_q(d,m)$, then $|V(F)|\le p_m -1$.
\end{cor}

\begin{proof}
This is an immediate consequence of Lemma \ref{lemp3}.
\end{proof}

\begin{lemma}\label{lemp4}
If $d\ge m(q-1)+1$, then $\prm_q(d,m) = \prm_q(m(q-1)+1,m) = \Fq^{{p^{}_m}}$. 
\end{lemma}
\begin{proof}
Let $\nu \in \{1, \dots , p^{ }_m\}$. Write  ${\mathsf{P}}_{\nu} \in \Fq^{{{m+1}}}$ as 
 ${\mathsf{P}}_{\nu} = (a^{ }_{0},\dots,a^{ }_{j_{\nu}-1},1,0,\ldots,0)$ for unique  
$j_{\nu}\in\{0,1,\dots , m\}$ and  $a^{ }_0,\ldots, a^{ }_{j_\nu-1}\in \Fq$. 
Now assume that  $d\ge m(q-1)+1$. 
Define the polynomial $F^{}_{\nu} = F^{}_{\nu}(X_0, \dots , X_m)$ by 
\begin{eqnarray}\label{eqprm:Fnu}
F^{}_{\nu} := X^{d-m(q-1)}_{j_\nu}\prod\limits_{i=0}^{j_\nu-1}\left( X_{j_\nu}^{q-1}
- (X_i-a^{ }_{i}X_{j_\nu})^{q-1}\right)\prod\limits_{k=j_\nu+1}^{m} (X_{j_\nu}^{q-1}  - X_k^{q-1}).
\end{eqnarray}
Clearly, $F^{}_{\nu} \in \Fq[X_0, \dots , X_m]_d$. 
Observe that 
$F_{\nu}({\mathsf{P}}_{\nu}) = 1$ and $F_{\nu}({\mathsf{P}}_{\mu}) =0$ for any $\mu \in \{1, \dots , p^{ }_m\}$ with $\mu \ne \nu$. Hence
any $\lambda = (\lambda_1, \dots , \lambda_{p^{ }_m}) \in \Fq^{p^{ }_m}$ can be written as 
$\lambda = \Ev(F)$, where 
$F = \lambda_{1} F_{1} + \cdots + \lambda_{p^{ }_m} F_{p^{ }_m}$. It follows that $\prm_q(d,m) = \Fq^{{p^{ }_m}}$.
\end{proof}

We remark that the polynomial \eqref{eqprm:Fnu} was used 
by S\o rensen in \cite[Lem. 3]{S}. 

\section{Minimum Distance of Projective Reed-Muller Codes}

In view of the results of the previous section, every codeword of $\prm_q(d,m)$ is of the form  $\text{Ev}(F)$ for a unique $F\in {C_q(d,m)}$, and determining the Hamming weight of this codeword corresponds to finding the number of zeros in $\PP^m(\Fq)$ of $F$. 

For any 
$f_1,\ldots,f_r\in \Fq[X_1,\ldots,X_m]_{\le d}$, we denote by $Z(f_1,\ldots,f_r)$ the set of common zeros of $f_1,\ldots,f_r$ in $\mathbb{A}^m(\Fq)$, whereas for $F_1,\ldots, F_r\in \Fq[X_0,\ldots,X_m]_{d}$, we denote by $V(F_1,\ldots,F_r)$ the set of common zeros of $F_1,\ldots,F_r$ in $\PP^m(\Fq)$.  By a \emph{hyperplane} in $\PP^m(\Fq)$ we shall mean a subset of the form $V(L)$, where $L$ is a nonzero homogeneous linear polynomial in 
$\Fq[X_0,\ldots,X_m]$.   We shall denote by $\widehat{\PP^m}$  the set of all hyperplanes in $\PP^m(\Fq)$. Note that the cardinality of $\widehat{\PP^m}$ is $p_m$. Also note that  if $P, P'$ are distinct points  of $\PP^m(\Fq)$, then $|\{H\in \widehat{\PP^m} : P\in H\}| =p_{m-1}$
and $|\{H\in \widehat{\PP^m} : P, P'\in H\}| =p_{m-2}$.

\begin{lemma}
\label{lem:newlem3}
Let $F$ be a nonzero projectively reduced polynomial in $\Fq[X_0,\ldots,X_m]_d$. 
Suppose $L$ is a nonzero homogeneous linear polynomial in 
$\Fq[X_0,\ldots,X_m]$ such that the hyperplane $V(L)$ is contained in $V(F)$. Then $L$ divides $F$. 
\end{lemma}

\begin{proof}
By a linear change of coordinates, we may assume that $L=X_m$. 
Separating the monomials in $F$ that are divisible by $X_m$, we can write 
$F = X_mG +R$ for some homogeneous and projectively reduced polynomials $G\in \Fq[X_0, \dots , X_m]_{d-1}$ and $R\in \Fq[X_0, \dots , X_{m-1}]_{d}$. Since $V(X_m) \subseteq V(F)$, we see that $R(\mathsf{P}) =0$ for all $\mathsf{P}\in \PP^{m-1}(\Fq)$. 
But $R$ is projectively reduced, and so it follows from Lemma~\ref{lemp3} that $R$ is the zero polynomial. Thus, $X_m$ divides $F$. 
\end{proof}

\begin{thm}\label{maxzero}
Let $F$ be a nonzero polynomial in $C_q(d,m)$. 
Then 
\begin{equation}
\label{eq:mainineq}
|V(F)|\leq p^{ }_m-\lceil{(q-s)q^{m-t-1}\rceil},
\end{equation}
where $t,s$ are unique integers such that $d-1=t(q-1)+s$ with $t\ge 0$ and $0\le s<q-1$. Moreover, if $1\le d\leq m(q-1)$ and if the equality holds in \eqref{eq:mainineq}, then $V(F)$ is a union of at most $d$ distinct hyperplanes in $\PP^m(\Fq) $.
\end{thm}

\begin{proof} 
In case $d\geq m(q-1)+1$, or equivalently, $t>m-1$, then $\lceil{(q-s)q^{m-t-1}\rceil}=1$, and so \eqref{eq:mainineq} is a consequence of Corollary~\ref{cor:lem1}.

Now assume that  $1\le d\leq m(q-1)$. Then  $t\leq m-1$ and therefore 
$$
p^{ }_m-\lceil{(q-s)q^{m-t-1}\rceil} \ge p^{ }_m-q^m >0.
$$
In particular, \eqref{eq:mainineq} holds trivially if $V(F)$ is empty. Thus we shall assume that $V(F)$ is nonempty.
 
We will use induction on $m$. 
First, suppose $m=1$. Then, $d\leq q-1$, and hence $t=0$. 
Consequently, 
$ p^{ }_m-\lceil{(q-s)q^{m-t-1}\rceil} = p^{ }_1-(q-s)q^{1-0-1}=s+1=d$. 
Thus 
\eqref{eq:mainineq} is a consequence of the elementary fact that a nonzero homogeneous polynomial of degree $d$ in two variables over a field $F$ has at most $d$ non-proportional zeros in $F^2\setminus \{(0,0)\}$.  Moreover since hyperplanes in $\PP^1(\Fq)$ are given by points, we also see that if  $m=1$ and if the equality holds in \eqref{eq:mainineq}, then $V(F)$ is a union of exactly $d$ distinct hyperplanes in $\PP^m(\Fq) $.
%

Next, suppose $m>1$ and the result holds for smaller values of $m$. 
Define
\begin{eqnarray*}
\Gamma:=\{H\in \widehat{\PP^m} : H\subseteq V(F)\} .
\end{eqnarray*}
In view of Lemma \ref{lem:newlem3}, we see that $\Gamma$ can contain at most $d$ distinct hyperplanes in $\PP^m(\Fq) $. Write
$$
\Gamma=\{H_1,\ldots,H_\ell\} \quad \text{and} \quad \EuScript{L}=\bigcup_{i=1}^\ell H_i
$$
for some nonnegative integer $\ell \le d$. 
We shall now divide the proof in two mutually exclusive and exhaustive cases. \\

\noindent
{\bf Case 1.} $\Ls=V(F)$.

In this case, we must have $\ell \ge 1$ since $V(F)$ is nonempty. By a linear change of
coordinates, we may assume that $H_1=V(X_m)$.
In view of Lemma~\ref{lem:newlem3}, we see that $F(X_0,\ldots,X_{m-1},0)=0$ and $f(X_0,\ldots,X_{m-1}):=F(X_0,\ldots,X_{m-1},1)$ is a nonzero reduced polynomial in $\Fq[X_0,\ldots,X_{m-1}]_{\le d-1}$. Thus
\begin{eqnarray}
|V(F)| &=& |V(F(X_0,\ldots,X_{m-1},0))|+|Z(f)|\nonumber\\
& \leq & p^{ }_{m-1}+q^m-\lceil{(q-s)q^{m-t-1}\rceil}\label{(2)}\\
&  = & p^{ }_m-\lceil{(q-s)q^{m-t-1}\rceil},\nonumber
\end{eqnarray}
where the second step uses the formula (recalled in \S\,\ref{subsec:RM}) for the minimum distance of Reed-Muller code
$\RM_q(d-1,m)$ and the fact that 
$\mathrm{ev}(f)$ is a nonzero codeword of this code since $f$ is nonzero reduced polynomial. Thus  \eqref{eq:mainineq} is proved in Case~1. 
\\

\noindent
{\bf Case 2.} $\Ls\neq V(F)$.

In this case, there exists some $P\in V(F)\setminus \Ls$. In particular, a hyperplane 
in $\PP^m(\Fq)$ passing through $P$ can not be contained in $V(F)$. For sections of 
$V(F)$ by such hyperplanes, we make the following \\

\noindent
{\bf Claim:} If $H$ is any hyperplane 
in $\PP^m(\Fq)$ such that $H\not\subseteq V(F)$, then 
$$
|V(F)\cap H|\leq p^{ }_{m-1}-\lceil{(q-s)q^{m-t-2}\rceil}.
$$

To prove this, fix a hyperplane $H$ 
in $\PP^m(\Fq)$ such that $H\not\subseteq V(F)$. By a linear change of coordinates, we
may assume that $H=V(X_m)$. As in the proof of Lemma~\ref{lem:newlem3}, we can write 
$F = X_mG +R$ for some homogeneous and projectively reduced polynomials $G\in \Fq[X_0, \dots , X_m]_{d-1}$ and $R\in \Fq[X_0, \dots , X_{m-1}]_{d}$. Since $V(X_m) \not \subseteq V(F)$, we see that $R$ is a nonzero polynomial in $C_q(d,m-1)$. Also,  
$V(F)\cap H = V(R)\cap H$ is the zero set of $R$ in $\PP^{m-1}(\Fq)$. Hence by the induction hypothesis,  $|V(F)\cap H|\leq p^{ }_{m-1}-\lceil{(q-s)q^{m-t-2}\rceil}$.
This proves the chaim. \\ 

%
Recall that we had chosen  $P\in V(F)\setminus \Ls$. 
Following Serre \cite{Se}, we consider the incidence set 
$$
\Ie:=\{(P',H): P'\in V(F)\setminus \{P\} \text{ and } H \in \widehat{\PP^m} \text{ with } P,P'\in H\}.
$$
This set can be counted in two different ways as follows. On the one hand,
$$
|\Ie |\; = \sum_{P'\in V(F)\setminus \{P\} } \sum_{\substack{H \in \widehat{\PP^m} \\ P,P'\in H}} 1 
\; =  \sum_{P'\in V(F)\setminus \{P\} } p^{ }_{m-2} 
\; =  \; p^{ }_{m-2} (|V(F)|-1).
$$
%
On the other hand, 
\begin{eqnarray*}
|\Ie| &=& \sum_{\substack{ H\in \widehat{\PP^m} \\ P\in H} } 
\sum_{ \substack{ P'\in V(F)\setminus \{ P\}  \\ P'\in H} } 1 \\
&=& \sum\limits_{\substack{ H\in \widehat{\PP^m}\\ P\in H} } \left( |V(F)\cap H|-1 \right) \\
&\le & \sum\limits_{\substack{ H\in \widehat{\PP^m}\\ P\in H} } \left(  p^{ }_{m-1}-\lceil{(q-s)q^{m-t-2}\rceil}-1 \right) \\
& =& p^{ }_{m-1} \left( p^{ }_{m-1}-\lceil{(q-s)q^{m-t-2}\rceil}-1 \right)
\end{eqnarray*}
where the inequality in Step 3 follows from the claim above. 
On comparing both these values of $|\Ie|$, 
we obtain
\begin{eqnarray*}
&|V(F)|&\leq \frac{p^{ }_{m-1}(p^{ }_{m-1}-\lceil{(q-s)q^{m-t-2}\rceil}-1)}{p^{ }_{m-2}}+1\\
& & = \frac{(p^{ }_{m-1})^2-p^{ }_{m-1}+p^{ }_{m-2}}{p^{ }_{m-2}}-\frac{p^{ }_{m-1}\lceil{(q-s)q^{m-t-2}\rceil}}{p^{ }_{m-2}}\\
& & = \frac{(p^{ }_{m-1})^2-q^{m-1}}{p^{ }_{m-2}}-\frac{p^{ }_{m-1}}{p^{ }_{m-2}}\lceil{(q-s)q^{m-t-2}\rceil}\\
& & = \frac{(q^m-1)^2-q^{m-1}(q-1)^2}{(q^{m-1}-1)(q-1)}-\left(\frac{q^{m}-1}{q^{m-1}-1}\right)\lceil{(q-s)q^{m-t-2}\rceil}\\
& & = \frac{q^{m+1}-1}{q-1}-\left(\frac{q^{m}-1}{q^{m-1}-1}\right)\lceil{(q-s)q^{m-t-2}\rceil}\\
& & < p^{ }_m-q\lceil{(q-s)q^{m-t-2}\rceil},
\end{eqnarray*}
where the last inequality follows since $q^m-1>q^{m}-q$. 
Thus,
\begin{eqnarray*}
|V(F)|&<& 
\begin{cases}
 p^{ }_m- (q-s)q^{m-t-1} & \text{ if } t<m-1, \\
 p^{ }_m-q &\text{ if } t= m-1
  \end{cases}
\\[1em]
&\le & 
\begin{cases}
 p^{ }_m-(q-s)q^{m-t-1} & \text{ if } t<m-1,\\
 p^{ }_m-(q-s) & \text{ if } t= m-1
 \end{cases}
\\[1em]
&=&p^{ }_m-\lceil{(q-s)q^{m-t-1}\rceil}.
\end{eqnarray*}
This yields \eqref{eq:mainineq}.  Moreover, since the ineequality in \eqref{eq:mainineq} is strict in Case 2,  the second assertion is proved as well. 
\end{proof}

\begin{cor}\label{mindis}
Assume that $1\le d\le m(q-1)+1$. Then the minimum distance of projective Reed-Muller code $\prm_q(d,m)$ is 
\begin{eqnarray*}
{(q-s)q^{m-t-1}},
\end{eqnarray*}
where $t,s$ are unique integers such that $d-1=t(q-1)+s$ with $t\ge 0$ and $0\le s< q-1$. Moreover, if $\omega_1,\ldots,\omega_s$ are any distinct elements of $\Fq$ and if 
\begin{equation}
\label{eq:MinimalG}
G:=X_t\prod_{i=0}^{t-1} (X_i^{q-1}-X_t^{q-1})\prod_{j=1}^s (X_{t+1}-\omega_j X_t),
\end{equation}
then $c^{ }_G =\mathrm{Ev}(G)$ is a minimum weight codeword of $\prm_q(d,m)$.
\end{cor}
\begin{proof}
Since $d\le m(q-1)+1$, we see that $t\le m$ and moreover, if $t=m$, then $s=0$. Thus,
$\lceil{(q-s)q^{m-t-1}\rceil} = {(q-s)q^{m-t-1}}$. Hence Theorem~\ref{maxzero} shows that 
\begin{equation}
\label{eq:MinDistBound}
d(\prm_q(d,m))\geq {(q-s)q^{m-t-1}}.
\end{equation}
Next, suppose $\lambda_1,\ldots,\lambda_s$ are distinct elements of $\Fq$ and $G$ is the polynomial given by \eqref{eq:MinimalG}. Clearly, $G\in \Fq[X_0,\ldots,X_m]_d$ and 
$G$ has zeros at all points of $\PP^m(\Fq)$ except those of the form $(0:\cdots:0:1:a_{t+1}:a_{t+2}:\cdots:a_m)$, where $a_{t+1}\neq \omega_j$ for $j=1,\ldots,s$ and $a_r\in \Fq$ for $r=t+2,\ldots,m$. Thus $|V(G)|=p^{ }_m-{(q-s)q^{m-t-1}}$. This, together with 
\eqref{eq:MinDistBound} shows that  the minimum distance of $\prm_q(d,m)$ is ${(q-s)q^{m-t-1}}$ and moreover, $c^{ }_G =\mathrm{Ev}(G)$ is a minimum weight codeword of $\prm_q(d,m)$.
\end{proof}

\begin{rmk}
{\rm The proof of minimum distance given by S\o rensen \cite[Thm.~1]{S} has a mistake in the case $t=m-1$, i.e., when $d-1 = (m-1)(q-1)+s$,  $0\le s < q-1$.  According to the argument in \cite[p. 1569]{S}, if $\PP^m(\Fq)\setminus V(F)=\{P_1,\ldots,P_e\}$,  then we can find linear homogeneous polynomials $G_i(X), i=1,\ldots, e-1$ such that $G_i(P_j)=\delta_{ij}$ for $i=1,\ldots,e-1,\ j=1,\ldots,e$.\footnote{If one is writing equalities such as $G_1(P_1)=1$, then one should take $P_1,  \dots , P_e$ to be certain fixed representatives in $\Fq^{m+1}$ of points of 
$\PP^m(\Fq)\setminus V(F)$, but this is tacitly understood in \cite{S} and stated explicitly in Section II of \cite{S}.  Also, unlike in \cite{S}, we have denoted $|\PP^m(\Fq)\setminus V(F)|$ by $e$ (rather than $t$) so as to avoid a conflict with the notation used in this paper. } However, this does not seem correct, since it is possible that some $P_j$ is a linear combination of the other points $P^{ }_i$.  Also, $H(X)=F(X)\prod_{i=1}^{e-1}G_i(X)$ does not imply that $V(H)=\PP^m(\Fq)\setminus \{P_e\}$. To correct this, we may use a weaker condition. 
Namely,  for each positive integer $ i\le e-1$, take $G_i(X)$ to be a  linear homogeneous polynomial such that $V(G_i)$ is a hyperplane which passes through $P_i$ but not through $P_e$. This  is possible since $P_i, P_e$ are distinct points of $\PP^m(\Fq)$. [See, for example,   Lemma~\ref{lem:LineLem} below.] Then $H(X):=F(X)\prod_{i=1}^{e-1}G_i(X)$ does have the property that 
 $V(H)=\PP^m(\Fq)\setminus \{P_e\}$.}
\end{rmk}

\begin{rmk}
{\rm 
An alternative proof of the formula for the minimum distance of $\prm_q(d,m)$ is given by 
Carvalho, Neumann, and Lopez \cite[Cor. 3.9 (c)]{CNL}, where they derive it from a more general result on projective nested cartesian codes.
}
\end{rmk}

\section{Minimum Weight Codewords of Projective Reed-Muller Codes}

In this section, we characterize the minimum weight codewords of projective Reed-Muller codes in terms of homogeneous polynomials whose images under the evaluation map \eqref{eqprm:Ev} are these codewords. We begin with a simple observation. 

\begin{lemma}
\label{lem:LineLem}
Given any $A, B\in \PP^m(\Fq)$ with $A\ne B$,  there is a homogeneous linear polynomial 
$L \in \Fq[X_0,\ldots,X_m]$ such that $L(A)=0$ and $L(B) \ne 0$. 
\end{lemma}

\begin{proof}
Let $A = (a_0:a_1: \cdots : a_m)$ and $B = (b_0:b_1: \cdots : b_m)$ be distinct elements of 
$ \PP^m(\Fq)$.  Since $A\ne B$,  there are integers $i, j$ with $0\le i < j \le m$ such that $a_ib_j - a_jb_i \ne 0$.  Now $L = a_i X_j - a_j X_i$ has the desired properties. 
\end{proof}

Let us recall that the codewords of $\prm_q(d,m)$ are of the form $c^{}_F$,  defined as in 
\eqref{eqprm:Ev},  where $F$ varies over the elements of $\Fq[X_0,\ldots,X_m]_d$.  The \emph{support} of such a codeword is given by the set
$$
\supp \left( c^{}_F \right) = \left\{  \mathsf{P} \in  \{\mathsf{P}_1,  ,\ldots, \mathsf{P}_{p^{ }_m}\} : 
F(\mathsf{P}) \ne 0\right\},
$$
where $\mathsf{P}_1,\ldots, \mathsf{P}_{p^{ }_m}$ are as in \S \ref{subsec:PRM}.
For any $F\in \Fq[X_0,\ldots,X_m]_d$,  we can 
clearly identify $\supp \left( c^{}_F \right)$  with $ \PP^m(\Fq) \setminus V(F)$.  In case $c^{}_F$ is a minimum weight codeword,  i.e.,  a nonzero codeword of least Hamming weight, 
we can say a little more.  

\begin{lemma}
\label{lem:Lem-cFcG}
Let $F\in \Fq[X_0,\ldots,X_m]_d$ be such that $c^{}_F$ is a minimum weight codeword of $\prm_q(d,m)$.  Then for any $G\in \Fq[X_0,\ldots,X_m]_d$,
$$
V(F) = V(G) \Longleftrightarrow 
\supp \left( c^{}_F \right) = \supp \left( c^{}_G \right)  
\Longleftrightarrow  c^{}_F = c^{}_{\lambda G} \text{ for some nonzero } \lambda \in \Fq.
$$
\end{lemma}

\begin{proof}
The first equivalence as well as the reverse implication in the second equivalence is obvious.  
Suppose $\supp \left( c^{}_F \right) = \supp \left( c^{}_G \right)$.  
Pick any  ${\mathsf{P}}_{\nu}\in \supp \left( c^{}_F \right)$. Then 
 $F({\mathsf{P}}_{\nu}) \ne 0 \ne G({\mathsf{P}}_{\nu})$, and so 
 there is $\lambda \in \Fq$ with $\lambda \ne 0$ such that 
$F({\mathsf{P}}_{\nu}) = \lambda \, G({\mathsf{P}}_{\nu})$.  Consequently, 
$\supp( c^{}_{F-\lambda G}) \subseteq \supp (c^{}_F) \setminus \{ {\mathsf{P}}_{\nu}\}$.
Since $c^{}_F$ is a minimum weight codeword,  it follows that $c^{}_{F-\lambda G} =0$, that is,   $c^{}_F = c^{}_{\lambda G}$.
\end{proof}


We are now ready to prove the main result of this section.

\begin{thm}\label{prm:charmin}
Assume that $1\le d\le m(q-1)+1$ and let $\, t , s$ be unique integers such that 
$d-1=t(q-1)+s$ with $t\ge 0$ and $0\le s< q-1$.
Then $c\in \Fq^{{p^{}_m}} $ is a minimum weight codeword of $\prm_q(d,m)$ if and only if $c=c^{}_Q$ for some 
$Q\in \Fq[X_0,\ldots,X_m]_d$ of the form
\begin{eqnarray}\label{eqprm:min}
Q= L_t\prod_{i=0}^{t-1} \left(L_t^{q-1}-L_i^{q-1}\right)\prod_{j=1}^{s}(L_{t+1}-\omega_j L_t)
\end{eqnarray}
where  
$L_0,\ldots, L_{t+1}$ are linearly independent homogeneous linear polynomials in  $\Fq[X_0,\ldots,X_m]$ and $\omega_1,\ldots,\omega_s$ are distinct elements of $\Fq$. [Here it is understood that if $s=0$, then the second product in \eqref{eqprm:min} is empty (and hence equal to $1$) and in this case only the existence and linear independence of $L_0, \ldots, L_{t}$ is asserted.]
\end{thm}
\begin{proof}

If $Q$ is as in \eqref{eqprm:min}, then it is clear that 
$Q\in \Fq[X_0,\ldots,X_m]_d$ and moreover, by a linear change of coordinates, 
we see that $|V(Q)|=|V(G)|$, where $G$ is as in \eqref{eq:MinimalG}. Thus Corollary \ref{mindis} implies that $c^{ }_Q$  is a minimum weight codeword of $\prm_q(d,m)$.

To prove the converse,  let $F\in \Fq[X_0,\ldots,X_m]_d$ be a nonzero polynomial such that $|V(F)|=p^{ }_m-(q-s)q^{m-t-1}$.  
We need to show that there exists some $Q\in \Fq[X_0,\ldots,X_m]_d$ of the form 
\eqref{eqprm:min} such that $c^{ }_F = c^{ }_Q$.  
To this end,  let us first consider the extreme case when $d = m(q-1)+1$ so that $t=m$ and $s=0$.  
In this case  $|V(F)|=p^{ }_m-1$ and so there is unique $\nu \in \{1, \dots , p^{ }_m\}$ such that 
$F({\mathsf{P}}_{\nu}) \ne 0$ and $F(\mathsf{P}_{\mu})=0$ for all $\mu \in \{1, \dots , p^{ }_m\}$ with $\mu \ne \nu$.   Thus if $F_{\nu}$ is as in \eqref{eqprm:Fnu},  then 
by Lemma~\ref{lem:Lem-cFcG}, $c^{ }_F = c^{ }_{\omega_0 F_{\nu}}$ 
for some $\omega_0\in \Fq$ with $\omega_0 \ne 0$.  Now,  if as in 
the proof of Lemma~\ref{lemp3},  we write  ${\mathsf{P}}_{\nu} = (a^{ }_{0},\dots,a^{ }_{j_{\nu}-1},1,0,\ldots,0)$,  and we let 
$$
L_t = L_m =   \omega_0 X_{j_{\nu}} \quad \text{and} \quad 
L_i = \begin{cases}  
(X_i-a^{ }_{i}X_{j_\nu}) & \text{if } 0\le i < j_{\nu}, \\
X_{i+1} & \text{if } j_{\nu}\le i \le m-1,
\end{cases}
$$
then $L_0, \dots , L_t$ are linearly independent and it is clear from \eqref{eqprm:Fnu} that $F$ 
is as in   \eqref{eqprm:min}. 

Thus,  we will now assume that $1\le d \le m(q-1)$.  Then $0\le t \le m-1$ and $0\le s \le q-2$. 
This implies that 
\begin{equation}
\label{eq:difftwo}
p^{ }_m-|V(F)| = (q-s)q^{m-t-1} \ge 2.
\end{equation}
Further, since $|V(F)|$ attains the bound in \eqref{eq:mainineq},  we see from Theorem~\ref{maxzero} that there is a positive integer $r\le d$ and homogeneous linear polynomials 
$G_1, \dots , G_r$ in $\Fq[X_0,\ldots,X_m]$ such that the corresponding hyperplanes $V(G_1), \dots , V(G_r)$ are distinct (i.e.,  $G_1, \dots , G_r$ are pairwise linearly independent) and 
$$
V(F) = \bigcup_{i=1}^r V(G_i).
$$
We claim that $r=d$. To see this,  assume the contrary,  i.e.,  suppose 
$1\le r < d$. 
In view of \eqref{eq:difftwo},  there are $A, B \in \PP^m(\Fq)\setminus V(F)$ such that $A\ne B$.  Hence by Lemma~\ref{lem:LineLem},  there is a homogeneous linear polynomial $L$ in $ \Fq[X_0, \dots , X_m]$ such that $L(A)=0$ and $L(B)\ne 0$.  Now consider $G:= L^{d-r}G_1\cdots G_r$.  Then
$G\in \Fq[X_0,\ldots,X_m]_d$ and $G(P) = 0$ for all $P \in V(F)\cup \{A\}$.  Also $c^{}_G$ 
is a nonzero codeword of $\prm_q(d,m)$ since $G(B)\ne 0$.  But the Hamming weight of $c^{}_G$ is smaller than that of $c^{}_F$ since $|V(G)| > |V(F)|$.  This contradicts the assumption that $c^{}_F$ is a minimum weight codeword.  Hence the claim that $r=d$ is proved. 
%

Now let $G:=G_1\cdots G_d$.  Then   $V(F) = V(G)$.   Hence by Lemma \ref{lem:Lem-cFcG}, 
there is $\lambda \in \Fq$ with $\lambda \ne 0$ such that 
$c^{}_F = c^{}_{\lambda G}$.  
%
We can thus assume without loss of generality that $F$ is a product of $d$ pairwise linearly independent homogeneous linear factors, and moreover by a linear change of coordinates, we may assume that one of the linear factors is $X_m$.  Consequently, we can write $F=X_mF_1$ for some $F_{1} \in  \Fq[X_0,\ldots,X_m]_{d-1}$ such that $F_1$ is not divisible by $X_m$. 
Thus,  if we let $f_1(X_0, \dots , X_{m-1}) = F_1(X_0, \dots , X_{m-1}, 1) $,  then 
$f_1$ is a polynomial in $\Fq[X_0,\ldots,X_{m-1}]$ of degree ${d-1}$.  Moreover, by intersecting $V(F)$ with the hyperplane $V(X_m)$ and its complement in $\PP^m(\Fq)$,  we see that
$$
p^{ }_m - (q-s)q^{m-t-1} = |V(F)| = p^{ }_{m-1} + |Z(f_1)|  \ \text{ and so } \
|Z(f_1)|  = q^m - (q-s)q^{m-t-1}.
$$
It follows that $\mathrm{ev}(f_1)$ can be viewed as a minimum weight codeword of the Reed-Muller code $\RM_q(d-1,m)$. Hence by \cite[Thm. 2.6.3]{DGM}, 
$\mathrm{ev}(f_1) = \mathrm{ev}(f)$ for some $f\in \Fq[X_0,\ldots,X_{m-1}]$ of the form \eqref{eq:minwtrm}, i.e., 
\begin{equation*}
f =  \omega_0\left(\prod_{i=1}^{t}(1-\ell_i^{q-1})\right)  \prod_{j=1}^{s} (\ell_{t+1}-\omega'_j)
\end{equation*}
for some linearly independent polynomials $\ell_1,\ldots,\ell_{t+1} \in \Fq[X_0,\ldots,X_{m-1}]$ of degree~$1$, and some  
$\omega_0, \omega'_1, \dots , \omega'_s \in \Fq$ 
with $\omega_0\neq 0$ and $\omega'_1, \dots , \omega'_s$ distinct.  Now 
$$
F(X_0,\ldots,X_m)=
X_m^{d}F_1\bigg(\frac{X_0}{X_m},\ldots,\frac{X_{m-1}}{X_m},1\bigg)\
= X_m\bigg(\! X_m^{d-1}f_1\bigg(\frac{X_0}{X_m},\ldots,\frac{X_{m-1}}{X_m}\bigg)\bigg)
$$
and so if we let 
$$
Q(X_0,\ldots,X_m):=X_m\left(X_m^{d-1}f\left(\frac{X_0}{X_m},\ldots,\frac{X_{m-1}}{X_m}\right)\right),
$$
then we see that $Q \in  \Fq[X_0,\ldots,X_m]_{d}$ and $F({\mathsf{P}}_{\nu})=Q({\mathsf{P}}_{\nu})$ for each  $\nu =1,\dots , p^{}_m$, and hence $c^{}_F = c^{}_Q$. Moreover if 
 we let 
$$
L_i(X_0,\ldots,X_m)
:=X_m\left(\ell_{i+1} \left(\frac{X_0}{X_m},\ldots,\frac{X_{m-1}}{X_m}\right)\right) \quad
\text{for $i=0,1,\dots , t-1$,}
$$
and also let $L_t := \omega_0 X_m$ and in case $s>0$,  we further let 
$$
L_{t+1}(X_0,\ldots,X_m)
:=X_m\left(\ell_{t+1} \left(\frac{X_0}{X_m},\ldots,\frac{X_{m-1}}{X_m}\right)\right) , 
$$
then we 
see that $L_1, \dots ,  L_{t+1}$ are nonzero homogeneous linear polynomials in 
$\Fq[X_0,\ldots,X_{m}]$  and moreover 
$$
Q= L_t\prod_{i=0}^{t-1} \left(L_t^{q-1}-L_i^{q-1}\right)\prod_{j=1}^{s}(L_{t+1}-\omega_j L_t)
\quad
\text{where} \quad \omega_j := \frac{\omega'_j}{\omega_0}  \text{ for } j=1,\dots , s. 
$$
It remains to show that $L_0, \dots ,  L_{t+1}$ are linearly independent, or more precisely,  
if $s=0$, then $L_0, \dots ,  L_{t}$ are linearly independent, 
whereas if $s>0$, then $L_0, \dots ,  L_{t+1}$ are linearly independent. 
To this end,  first suppose
$s=0$.   Since $\ell_1,\ldots,\ell_{t}$ are linearly independent, 
it is clear that $L_0,\ldots,L_{t-1}$ are linearly independent.  Now if $L_0, \dots ,  L_{t}$ were  linearly dependent,  then $X_m$ would be a $\Fq$-linear combination of $L_0, \dots ,  L_{t-1}$.
On the other hand, since $|V(Q)|= |V(F)|= p^{ }_m - q^{m-t}  < p^{ }_m $ , there is some 
$A\in \PP^m(\Fq)$ such that $Q(A)\ne 0$.  In particular,  $L_t(A)\ne 0$, but 
$L_i(A) =0$ for $i=0,1,\dots , t-1$,  which shows that $X_m$ cannot be a $\Fq$-linear combination of $L_0, \dots ,  L_{t-1}$.  

Next, suppose $s>0$.  As before,  the linear independence of  $\ell_1,\ldots,\ell_{t+1}$ shows that $L_0 , \dots ,  L_{t-1}, L_{t+1}$ are linearly independent.    Fix an arbitrary element $A$ of 
$\PP^m(\Fq) \setminus V(Q)$.  Since $Q(A)\ne 0$, 
we must have $L_t(A)\ne 0$ and 
$L_i(A) =0$ for $i=0,1,\dots , t-1$.  In particular, we can find a unique representative of $A$ in 
$\Fq^{m+1}$,  say $\mathsf{A}= (\mathsf{a}_0, \mathsf{a}_1, \dots ,\mathsf{a}_m)$,  such that $\mathsf{a}_m=1$.
 Suppose,  if possible,
$X_m$ is  a $\Fq$-linear combination of $L_0, \dots ,  L_{t-1},  L_{t+1}$.  Then  $L_{t+1}(\mathsf{A})$ is necessarily equal to a fixed nonzero scalar in $\Fq$,  say $\lambda$, which is independent of the choice of $A$ in $\PP^m(\Fq) \setminus V(Q)$.
%
 It follows that  $|\PP^m(\Fq) \setminus V(Q)|$ is at most 
$$
| \{ \mathbf{a}=(a_0, \dots , a_{m-1}) \in \Fq^m : 
\ell_i (\mathbf{a})=0 \text{ for } i=1,\dots ,t \text{ and } 
   \ell_{t+1}(\mathbf{a})=\lambda \} | = q^{m-t-1}, 
   $$
    where the last equality follows since $\ell_1,\ldots,\ell_{t+1}$ are linearly independent polynomials of degree $1$.  But then 
    $$
(q-s)q^{m-t-1}  =    |\PP^m(\Fq) \setminus V(Q)| \le q^{m-t-1},
$$
which is a contradiction since $s < q-1$.  This shows that $L_0, \dots ,  L_{t+1}$ are linearly independent. 
\end{proof}

\section{Enumeration of Minimum Weight Codewords}
\label{sec:enumMinWt}

In this section,  we will use Theorem~\ref{prm:charmin} to prove the following.

\begin{thm}\label{thm:enum}
Assume that $1\le d\le m(q-1)+1$ and let $t, s$ be unique integers such that 
$d-1=t(q-1)+s$ with $t\ge 0$ and $0\le s< q-1$.
Then the number of  minimum weight codewords of $\prm_q(d,m)$ is equal to 
\begin{equation}\label{NoMinWtPRM}
\left( q^{m+1} - 1 \right){{m}\brack{t}}_q N_s, 
\quad \text{where} \quad 
N_s=\begin{cases}
 1 & {\rm if}\ s=0,\\
\displaystyle{\frac{1}{s+1}\binom{q}{s} {\stirling{m-t}{1}}_q} & {\rm if}\ s\ne 0.
 \end{cases}
\end{equation}
\end{thm}

Before giving a proof of this result, it will be convenient to set some notation and terminology, and also some preliminary results which will be useful to us in the sequel. First, let us fix $d, t$ and $s$ as in the statement of Theorem~\ref{thm:enum}. 
We shall denote by $\M_q(d,m)$, and often simply by $\M$, the set of all  minimum weight codewords of $\prm_q(d,m)$. The polynomial ring $\Fq[x_0, \dots, x_m ]$ in $m+1$ variables with coefficients in $\Fq$ will be denoted by $R$, and for any nonnegative integer $e$, we denote by $R_e$ the space of homogeneous polynomials in $R$ of degree $e$, including the zero polynomial. In particular, $R_1$ is the $(m+1)$-dimensional $\Fq$-vector space of homogeneous linear polynomials in $R$. 
We shall use a slight variant of the notion of support of a codeword of $\prm_q(d,m)$. Namely, given any $F\in R_d$, we let
$$
\ssupp (c^{ }_F) = \{P\in \PP^m(\Fq) : F(P)\ne 0\} = \PP^m(\Fq) \setminus V(F).
$$
Thus, the only difference between $\ssupp (c^{ }_F)$ and $\supp (c^{ }_F)$ is that the former comprises of points of 
$\PP^m(\Fq)$ rather than their chosen representatives in $\Fq^{m+1}$. Let $\mathcal{P}(\PP^m)$ denote the set of all subsets of $\PP^m(\Fq)$. 
In view of Lemma~\ref{lem:Lem-cFcG}, we see that the map 
$\Phi : \M \to \mathcal{P}(\PP^m)$ defined by $ \Phi(c):= \ssupp (c)$
has the property that 
\begin{equation}\label{eq:MphiM}
	|\M| = (q-1) |\Phi(\M)|.
\end{equation}
Continuing with notations, for any $e\in \Z$ with 
$-1\le e\le m$, we denote by $\G_e(\PP^m)$ the Grassmannian consisting of all $e$-dimensional projective linear subspaces\footnote{As per the standard conventions, the empty subset of $\PP^m(\Fq)$, which corresponds to the zero subspace of $\Fq^{m+1}$, is the only projective linear subspace of $\PP^m(\Fq)$ of dimension $-1$.} of $\PP^m(\Fq)$. It is elementary and well-known that the cardinality of $\G_e(\PP^m)$ is the Gaussian binomial coefficient ${{m+1}\brack {e+1}}_q$. Here is a useful preliminary result. 

\begin{lemma}\label{lem:thetaTrick}
	Let 
	$K_0, \dots , K_r \in R_1$ be linearly independent. Then the map 
	$$
	\theta : \Fq^{m+1} \to \Fq^{r+1}   \quad \text{defined by} \quad \theta(\mathbf{a}):= \left(K_0(\mathbf{a}),  \dots , K_r(\mathbf{a}) \right)
	$$
	is surjective. 
\end{lemma}

\begin{proof}
	Extend $\{K_0,  \dots , K_r\}$ to a $\Fq$-vector space basis $\{K_0,  \dots , K_m\}$ of $R_1$. Then the map 
$	\widehat{\theta} : \Fq^{m+1} \to \Fq^{m+1}$ defined by $\widehat{\theta}(\mathbf{a}):= (K_0(\mathbf{a}),  \dots , K_m(\mathbf{a}))$ is an $\Fq$-vector space  isomorphism. In particular, $\widehat{\theta}$ is surjective, and hence so is $\theta$. 
\end{proof}

The following lemma will prove Theorem~\ref{thm:enum} in the special case when $s=0$.

\begin{lemma}\label{lem:EnumZero-s}
	Assume that $s=0$ so that $d= t(q-1)+1$ with $0\le t \le m$. Then 
	$$
	|\M_q(d,m)| = \left( q^{m+1} - 1 \right){{m}\brack{t}}_q.
	$$
\end{lemma} 

\begin{proof} 
	If $t=0=s$, then 
	$\prm_q(d,m)$ is the simplex code of length $p^{}_m$, dimension $m+1$ and minimum distance $q^m$. In this case, every nonzero codeword is a minimum weight codeword and so the desired formula clearly holds. We now assume that $t\ge 1$. Consider the incidence set $\mathcal{I}$ and the map $\tau : \mathcal{I} \to \Phi(\M)$ defined by 
	$$
	\mathcal{I}:=\{(E,H)\in \G_{m-t}(\PP^m)\times \G_{m-t-1}(\PP^m): H\subseteq E\}
	\quad \text{and} \quad \tau(E,H):= E \cap H^c, 
	$$
	where by $H^c$ we denote the complement of $H$ in $ \PP^m(\Fq)$, i.e., $H^c:= \PP^m(\Fq) \setminus H$. 
	Note that if $(E,H) \in \mathcal{I}$, then there are linearly independent $L_0, \dots , L_t \in R_1$ 
	 such that 
	$E = V(L_0, \dots , L_{t-1})$ and $H=V(	L_0, \dots , L_t)$. Consequently,
	$$
E\cap H^c = \ssupp (c^{ }_F)\quad \text{where} \quad 
F:=L_t\prod_{i=0}^{t-1} \left(L_t^{q-1}-L_i^{q-1}\right). 
$$
Thus $\tau$ is well-defined. Moreover, Theorem~\ref{prm:charmin} shows that $\tau$ is surjective. We claim that $\tau$ is also injective. To see this,  let $(E,H), \, (E',H')\in \mathcal{I}$ be such that $\tau(E,H)= \tau(E',H')$. Suppose, if possible,  $H\ne H'$. 
Then there are linearly independent $L_0, \dots , L_t \in R_1$ and 
linearly independent $L'_0, \dots , L'_t$ in $R_1$  such that 
$E = V(L_0, \dots , L_{t-1})$, $H=V(	L_0, \dots , L_t)$, 
$E' = V(L'_0, \dots , L'_{t-1})$, and $H'=V(	L'_0, \dots , L'_t)$. 
Since $H, H'$ have the same dimension and $H\ne H'$, we must have $H\not\subseteq H'$. Hence there is some $j\in \{0,\dots , t\}$ such that $L'_j$ does not vanish on $H$. 
This implies that $L_0, \dots , L_t, L_j'$ are linearly independent. First, suppose $0\le j \le t -1$. By Lemma~\ref{lem:thetaTrick}, there exists  $\mathbf{a} \in \Fq^{m+1}$ such that 
$L_0(\mathbf{a} )= \cdots = L_{t-1}(\mathbf{a} )=0$ and $L_t(\mathbf{a} )=1 = L'_j(\mathbf{a} )$. Next, suppose $j=t$. Again, 
by Lemma~\ref{lem:thetaTrick}, there exists  $\mathbf{b} \in \Fq^{m+1}$ such that 
$L_0(\mathbf{b} )= \cdots = L_{t-1}(\mathbf{b} )=0$ and $L_t(\mathbf{b} )=1$, but $L'_t(\mathbf{b} )=0$. 
Note that these $\mathbf{a}$ and $\mathbf{b}$, when they exist, are necessarily nonzero. Thus, 
in either case, if $P$ is the point of $ \PP^m(\Fq)$ corresponding to $\mathbf{a}$ or $\mathbf{b}$  (according as $0\le j<t$ or  $j=t$), then 
$P\in E\cap H^c$, but $P\not\in E'\cap H'^c$. 
But this contradicts the assumption that $\tau(E,H)= \tau(E',H')$. So we must have $H=H'$, and therefore $E=E'$. Thus we have shown that $\tau$ is bijective. Consequently, in view of \eqref{eq:MphiM}, we obtain 
$$
|\M_q(d,m)| = (q-1)|\mathcal{I}|=(q-1){\stirling{m+1}{m-t+1}}_q{\stirling{m-t+1}{1}}_q=\left( q^{m+1} - 1 \right){{m}\brack{t}}_q,
$$
as desired.
\end{proof}

Next, we consider the case when $s\ne 0$. By Theorem~\ref{prm:charmin}, we know that a minimum weight codeword of $\prm_q(d,m)$ is of the form $c^{ }_F$, where
\begin{equation}\label{eq:TypicalF}
F= L_t\prod_{i=0}^{t-1} \left(L_t^{q-1}-L_i^{q-1}\right)\prod_{j=1}^{s}(L_{t+1}-\omega_j L_t)
\end{equation}
for some linearly independent $L_0, \dots ,  L_{t+1} \in R_1$ and some distinct $\omega_1, \dots , \omega_s \in \Fq$. The support of such a codeword is given by $V(L_t)^c \cap V(L_0, \dots , L_{t-1}) \cap W$, where 
$$
W =   \bigcap_{j=1}^s V(L_{t+1}-\omega_j L_t)^c 
=  \bigcup_{\mu} V(L_{t+1}-\mu L_t),
$$
where  $\mu$ varies over elements of $\Fq \setminus \{\omega_1, \dots , \omega_s \}$. It is, however, possible that a different choice $L'_0, \dots ,  L'_{t+1} \in R_1$ could give rise to the same support. In effect, one can replace $L_t$ by a linear combination of the form $L_t + \sum_{i=0}^{t-1} c_iL_i$ or one can even swap it with a factor of the form $(L_{t+1}-\omega_j L_t)$.  Thus one has to be a little more careful while counting the possible supports of minimum weight codewords. 
To do it systematically, we consider a set $\mathcal{J}$ and a 
map 	$\Psi$ defined as follows. 

Let $\mathcal{P}_s(\Fq)$ denote the set of all subsets of $\Fq$ containing $s$ elements, and let 
$$
\mathcal{J}:=\left\{(E,L_t,L_{t+1},S)\in \Lambda 
:E\not\subseteq V(L_t) \text{ and } E\cap V(L_t)\not\subseteq V(L_{t+1})\right\},
$$
where $\Lambda := \G_{m-t}(\PP^m) \times R_1  \times R_1 \times \mathcal{P}_s(\Fq)$. Define 
$\Psi:\mathcal{J}\longrightarrow \Phi(\mathcal{M})$	 by 
$$
\Psi(E,L_t,L_{t+1},S):= (E\setminus V(L_t))\cap \Big(\! \bigcup_{\mu\in {\Fq\setminus S}} \! V(L_{t+1}-\mu L_t) \Big) \ 
\text{ for } (E,L_t,L_{t+1},S)\in 	\mathcal{J}.
$$
Note that if $(E,L_t,L_{t+1},S)\in 	\mathcal{J}$, then there are polynomials $L_0, \dots , L_{t-1} \in R_1$ such that 
$E= V(L_0, \dots , L_{t-1})$ and $L_0, \dots , L_{t+1}$ are linearly independent, and moreover, $	\Psi(E,L_t,L_{t+1},S)=\ssupp(c^{ }_F)$, 	where $F$ is given by \eqref{eq:TypicalF}. 
This shows that $\Psi$ is well-defined. Furthermore,  Theorem \ref{prm:charmin} shows that $\Psi$ is surjective. 
The next few lemmas will help us analyze the fibres of $\Psi$. 

\begin{lemma}
\label{lem:PsiLem1}
Let  $(E,L_t,L_{t+1},S)$ and $(E',L'_t,L'_{t+1},S')$ be elements of $\mathcal{J}$ such that 
$\Psi (E,L_t,L_{t+1},S) = \Psi(E',L'_t,L'_{t+1},S')$. Then $E=E'$. 
\end{lemma} 

\begin{proof}
	By our hypothesis, we can find $L_0, \dots , L_{t-1},  L'_0, \dots , L'_{t-1} \in R_1$ such that 
	$E= V(L_0, \dots , L_{t-1})$, $E'= V(L'_0, \dots , L'_{t-1})$, 
	and the subsets $\{L_0, \dots , L_{t+1}\}$  and $\{L'_0, \dots , L'_{t+1} \}$ of $R_1$ 
	are linearly independent; moreover, if $F$ is given by \eqref{eq:TypicalF} and $F'$ is given by \eqref{eq:TypicalF} with $L_i$ replaced by $L'_i$ for $0\le i \le t+1$ and $S$ by $S'$, then 
\begin{equation}\label{eq:FFprime}
\Psi (E,L_t,L_{t+1},S) =	\ssupp(c^{}_F)=\ssupp(c^{}_{F'}) = \Psi(E',L'_t,L'_{t+1},S').
\end{equation}
	Now suppose, if possible, $E\ne E'$. Then one among $L'_0, \dots , L'_{t-1}$, say $L'_0$,  must be such that 
	$L_0, \dots , L_{t-1},	L'_0$ are linearly independent. In case $L_0, \dots , L_{t+1},	L'_0$ are also linearly independent, then by Lemma~\ref{lem:thetaTrick}, we can find $P\in E$ such that $L_t(P)\ne 0$ and $L_{t+1}(P)= \mu L_t(P)$ for some $\mu \in \Fq\setminus S$ and moreover, $L'_0(P)\ne 0$. But then $P \in  \ssupp(c^{}_F)\setminus \ssupp(c^{}_{F'})$, which contradicts \eqref{eq:FFprime}. Thus $L_0, \dots , L_{t+1},	L'_0$ are  linearly dependent, and so we can find $c_0, \dots , c_{t+1}\in \Fq$ such that
	$$
	L'_0 =c_0L_0+\ldots+c_{t-1}L_{t-1}+c_tL_t+c_{t+1}L_{t+1}  \quad \text{and} \quad c_t, c_{t+1} \text{ are not both zero.}
	$$
		Now since $|S|=s\le q-2$, there exist distinct $\mu_1,\mu_2\in \Fq\setminus S$. 
		By   Lemma~\ref{lem:thetaTrick}, we can find $P_1,P_2\in E$ such that $L_t(P_i)\ne 0$ and $L_{t+1}(P_i)=\mu_iL_t(P_i)$ for $i=1,2$. 
		It follows that $P_1,P_2\in \ssupp(c^{}_F)=\ssupp(c^{}_{F'})$, and hence $L'_0(P_i)=0$ for $i=1,2$.  On the other hand, $L'_0(P_i)=(c_t+c_{t+1}\mu_i)L_t(P_i)$ for $i=1,2$. 
But then $c_t+c_{t+1}\mu_1=0 = c_t+c_{t+1}\mu_2$,  which implies that $c_{t+1}=0$ (since $\mu_1\ne \mu_2$), and hence $c_t=0$. Thus we obtain a contradiction. This proves that $E=E'$. 
\end{proof}

\begin{lemma}
	\label{lem:PsiLem2}
Suppose $(E,L_t,L_{t+1},S) \in \mathcal{J}$ and $S'\in \mathcal{P}_s(\Fq)$ have the property that 
	$\Psi (E,L_t,L_{t+1},S) = \Psi(E,L_t,L_{t+1},S')$.
Then $S=S'$. 
\end{lemma} 

\begin{proof}
Suppose, on the contrary, $S\ne S'$. Pick some $\omega\in S$ such that $\omega\not \in S'$. 	By   Lemma~\ref{lem:thetaTrick}, we
can find $P\in E$ such that  $L_t(P)\ne 0$ and $L_{t+1}(P)= \omega L_t(P)$. Then 
$P\in  \Psi(E,L_t,L_{t+1},S')$, while $P \not \in \Psi (E,L_t,L_{t+1},S) $, which is a contradiction. 
\end{proof}

\begin{lemma}
		\label{lem:PsiLem3}
		Let $(E,L_t,L_{t+1},S) \in \mathcal{J}$. Then there are exactly $q^{t+1} (q-1)$ elements in $\mathcal{J}$ of the
		form $(E,L_t,L'_{t+1},S')$ such that $\Psi (E,L_t,L_{t+1},S) = \Psi(E,L_t,L'_{t+1},S')$.
	\end{lemma} 

\begin{proof}
Suppose $L'_{t+1} \in R_1$ and $S'\in \mathcal{P}_s(\Fq)$ are such that 
$\Psi (E,L_t,L_{t+1},S) = \Psi(E,L_t,L'_{t+1},S')$. Choose $L_0, \dots , L_{t-1}  \in R_1$ such that 
$E= V(L_0, \dots , L_{t-1})$ and $\{L_0, \dots , L_{t+1}\}$ as well as $\{L_0, \dots , L_t, L'_{t+1}\}$ are 
 linearly independent. In case $\{L_0, \dots , L_{t+1}, L'_{t+1}\}$ is also linearly independent, then  
by Lemma~\ref{lem:thetaTrick}, we can find $P\in E$ such that  
$L_t(P)\ne 0$, $L_{t+1}(P)= \mu L_t(P)$, and $L'_{t+1}(P)= \omega L_t(P)$ 
 for some $\mu \in \Fq\setminus S$ and $\omega \in S$. But then 
$P\in  \Psi(E,L_t,L_{t+1},S') \setminus \Psi (E,L_t,L'_{t+1},S')$, which is a contradiction. Thus 
there are  $c_0, \dots , c_{t+1}\in \Fq$ such that
\begin{equation}\label{lincombForPsiLem3}
L'_{t+1} =c_0L_0+\ldots+c_{t-1}L_{t-1}+c_tL_t+c_{t+1}L_{t+1} .
\end{equation}
Note that $c_0, \dots , c_{t+1}$ are uniquely determined by $L'_{t+1}$ since $\{L_0, \dots , L_{t+1}\}$ is linearly independent. Moreover, $c_{t+1}\ne 0$ since $\{L_0, \dots , L_t, L'_{t+1}\}$ is also  
linearly independent. On the other hand, if $c_0, \dots , c_{t+1}\in \Fq$ are picked arbitrarily with $c_{t+1}\ne 0$, and 
if $L'_{t+1}$ is given by  \eqref{lincombForPsiLem3}, then we will show that there exists  unique $S'$ in $\mathcal{P}_s(\Fq)$ 
such that $\Psi (E,L_t,L_{t+1},S) = \Psi(E,L_t,L'_{t+1},S')$. To prove the existence, note that for any $\lambda\in \Fq$, 
$$
L'_{t+1} - \lambda L_t = c_tL_t+c_{t+1}L_{t+1}  - \lambda L_t = c_{t+1}\left( L_{t+1} - \frac{\lambda - c_t}{c_{t+1}}L_t\right) \quad \text{on points of } E.
$$ 
Thus if we take $S'\in \mathcal{P}_s(\Fq)$ such that $\Fq \setminus S'= \{\mu c_{t+1} + c_{t} : \mu \in \Fq \setminus S\}$, then 
$$
E \cap \Big( \bigcup_{\mu'\in {\Fq\setminus S'}} \! V(L'_{t+1}-\mu' L_t) \Big) = E \cap \Big( \bigcup_{\mu\in {\Fq\setminus S}} \! V(L_{t+1}-\mu L_t) \Big). 
$$
This implies that $\Psi (E,L_t,L_{t+1},S) = \Psi(E,L_t,L'_{t+1},S')$. The uniqueness of $S'$ (for the given choice of $L'_{t+1}$) is an immediate consequence of Lemma~\ref{lem:PsiLem2}. Since $L'_{t+1}$ of the form \eqref{lincombForPsiLem3} can be chosen in exactly $q^{t+1} (q-1)$ ways, the lemma is proved. 
\end{proof}

\begin{lemma}
	\label{lem:PsiLem4}
	Let $(E,L_t,L_{t+1},S) \in \mathcal{J}$. Then there are exactly $(s+1)(q-1)^2q^{2t+1} $ elements $(E',L'_t,L'_{t+1},S')$
		in $\mathcal{J}$ such that $\Psi (E,L_t,L_{t+1},S) = \Psi(E',L'_t,L'_{t+1},S')$.
\end{lemma} 

\begin{proof}
	Write $\sigma := \Psi (E,L_t,L_{t+1},S)$.   Choose $L_0, \dots , L_{t-1}  \in R_1$ such that 
	$E= V(L_0, \dots , L_{t-1})$ and $\{L_0, \dots , L_{t+1}\}$  is linearly independent. 
	Now if  $(E',L'_t,L'_{t+1},S')$ is in $\Psi^{-1}(\sigma)$, i.e., it is an element of $\mathcal{J}$ such that $\Psi(E',L'_t,L'_{t+1},S') = \sigma$, then by  Lemma~\ref{lem:PsiLem1}, $E'=E$, and therefore, $\{L_0, \dots , L_{t-1}, L'_t, L'_{t+1}\}$ is also 
	linearly independent. Further, 
	as before, 
	we can use   Lemma~\ref{lem:thetaTrick} to show that $\{L_0, \dots , L_{t+1}, L'_{t}\}$ is linearly dependent. Thus exactly one of the following two cases are possible. 

\medskip

\noindent
{\bf Case 1.} 
$L'_t:= b_0L_0+\dots+b_tL_t$ for some $b_0, \dots , b_{t}\in \Fq$. 

\medskip

In this case, we must have $b_t \ne 0$. Further, given any $b_0, \dots , b_{t}\in \Fq$ with $b_t \ne 0$, if we take 
$L'_t:= b_0L_0+\dots+b_tL_t$, then we can show that there is a unique $S^* \in \mathcal{P}_s(\Fq)$ for which 
$\Psi (E,L_t,L_{t+1},S) = \Psi(E',L'_t,L_{t+1},S^*)$. Indeed, the uniqueness of $S^*$ is clear from Lemma~\ref{lem:PsiLem2}, 
whereas for the existence, it suffices to observe that 
 for any $\lambda\in \Fq$, 
$$
L_{t+1} - \lambda L'_t = L_{t+1}  - \lambda b_t L_t  \quad \text{on points of } E, 
$$ 
and so we can take $S^*= \{ \mu/b_t :  \mu \in \Fq \setminus S\}$. 
It follows that there are exactly $q^t (q-1)$ choices of $b_0, \dots , b_{t}\in \Fq$ with $b_t \ne 0$ and a corresponding unique choice of $S^*$ for which $L'_t:= b_0L_0+\ldots+b_tL_t$ satisfies $\Psi (E,L_t,L_{t+1},S) = \Psi(E',L'_t,L_{t+1},S^*)$. Furthermore, for any such choice of $L'_t$ and $S^*\in \mathcal{P}_s(\Fq)$, by Lemma~\ref{lem:PsiLem3}, there are exactly $q^{t+1} (q-1)$ choices of $L'_{t+1}$ and 
corresponding unique choice of $S' \in \mathcal{P}_s(\Fq)$ for which $\Psi(E',L'_t,L_{t+1},S^*) = \Psi(E',L'_t,L'_{t+1},S')$. It follows that there are exactly $q^t (q-1) \times q^{t+1} (q-1) = q^{2t+1} (q-1)^2$ elements $(E',L'_t,L'_{t+1},S')$ in $\Psi^{-1}(\sigma)$ for which 
$L'_t:= b_0L_0+\ldots+b_tL_t$ for some $b_0, \dots , b_{t}\in \Fq$ with $b_t \ne 0$.

\medskip

\noindent
{\bf Case 2.} $L'_{t} =b_0L_0+\dots + b_{t+1}L_{t+1}$ for some $b_0, \dots , b_{t+1}\in \Fq$ with $b_{t+1} \ne 0$. 

\medskip

In this case, we claim that $-b_t/b_{t+1} \in S$. To prove this, assume the contrary. Use Lemma~\ref{lem:thetaTrick} to choose $P\in E$ such that $L_t(P) \ne 0$ and $L_{t+1}(P) = \left(-b_t/b_{t+1} \right) L_t(P)$. Then $P\in \sigma$, since 
$-b_t/b_{t+1} \not\in S$. On the other hand, our choice of $P$ implies that 
$$
L'_t(P) = b_t L_{t}(P) + b_{t+1}L_{t+1}(P) = 0
$$
and therefore $P\not\in \Psi(E',L'_t,L'_{t+1},S') = \sigma$, which is a contradiction.
Thus the claim is proved. We shall now show that for any fixed choice of $b_0, \dots , b_{t+1}\in \Fq$ with $b_{t+1} \ne 0$ and $-b_t/b_{t+1} \in S$, if we take $L'_{t} =b_0L_0+\dots + b_{t+1}L_{t+1}$,  and  $L^*_{t+1}:= L_t $, 
then  there is a unique $S^* \in \mathcal{P}_s(\Fq)$ such that $ (E,L'_t,L^*_{t+1},S^*) \in \mathcal{J}$ and 
$\Psi(E,L'_t,L^*_{t+1},S^*) = \sigma$. To see this, first note that $\{L_0, \dots ,L_{t-1}, L'_t,  L^*_{t+1}\}$  is linearly independent. 
Further, if we take 
$S^* \in \mathcal{P}_s(\Fq)$ such that
$$
\Fq\setminus S^* = \left\{ \frac{1}{b_t + \mu b_{t+1}} : \mu \in \Fq \setminus S\right\}.
$$
then $S^*$  is well-defined since $-b_t/b_{t+1} \in S$. Moreover, if for any $\mu^* \in \Fq\setminus S^*$, 
we write $\mu^* = 1/(b_t + \mu b_{t+1})$ for some $\mu \in  \Fq \setminus S$, then on points of $E$, we have 
$$
L^*_{t+1} - \mu^* L'_t = L_t - \frac{1}{b_t + \mu b_{t+1}} (b_tL_t +b_{t+1}L_{t+1})
= \frac{-b_{t+1}}{b_t + \mu b_{t+1}} \left( L_{t+1} - \mu L_t\right).
$$
Consequently, 
$$
(E\setminus V(L'_t))\cap \Big(\! \bigcup_{\mu^*\in {\Fq\setminus S^*}} \! V(L^*_{t+1}-\mu^* L'_t) \Big) = 
(E\setminus V(L_t))\cap \Big(\! \bigcup_{\mu\in {\Fq\setminus S}} \! V(L_{t+1}-\mu L_t) \Big) = \sigma. 
$$
This shows the existence of a polynomial $L^*_{t+1} \in R_1$ and a set $S^* \in \mathcal{P}_s(\Fq)$ such that $ (E,L'_t,L^*_{t+1},S^*) \in \mathcal{J}$ and 
$\Psi(E,L'_t,L^*_{t+1},S^*) = \sigma$. The uniqueness of $S^*$ for the given choice of $L^*_{t+1}$ is clear from 
  Lemma~\ref{lem:PsiLem2}. Moreover, by Lemma \ref{lem:PsiLem3}, there are exactly $q^{t+1} (q-1)$ choices of $L'_{t+1}$ and 
  corresponding unique choice of $S' \in \mathcal{P}_s(\Fq)$ for which $\Psi(E',L'_t,L^*_{t+1},S^*) = \Psi(E',L'_t,L'_{t+1},S')$.
  Since 
  $b_0, \dots , b_{t+1}\in \Fq$ such that $b_{t+1} \ne 0$ and $-b_t/b_{t+1} \in S$ can be chosen in $s(q-1) q^t$ ways, we see that there are exactly $q^{t+1} (q-1) \times s(q-1) q^t 
  = sq^{2t+1} (q-1)^2$ elements $(E',L'_t,L'_{t+1},S')$ in $\Psi^{-1}(\sigma)$ for which  $L'_t$ is as in Case 2.
 
 By combining the counts in Case 1 and Case 2, we obtain the desired result. 
\end{proof}

We are now ready to prove the main result of this section. 

\medskip

\noindent{\emph{Proof of Theorem~\ref{thm:enum}.}  If $s=0$, then the desired formula \eqref{NoMinWtPRM} follows from 
	Lemma~\ref{lem:EnumZero-s}. If $s\ne 0$, then by Lemma~\ref{lem:PsiLem4}, 
\begin{equation}\label{eq:FibreCard}
	|\Psi^{-1}(\sigma)|=(s+1)(q-1)^2 q^{2t+1} \quad \text{for all }\sigma\in \Phi(\mathcal{M}).
\end{equation}
It follows that 
\begin{align*}
|\Phi(\mathcal{M})| & = \frac{|	\mathcal{J}|}{ (s+1)(q-1)^2 q^{2t+1}} \\
&= \frac{1}{(s+1)(q-1)^2 q^{2t+1}} {{m+1}\brack{m-t+1}}_q \left( q^{m+1} - q^t \right) \left( q^{m+1} - q^{t+1} \right){ q \choose s} \\
&= \frac{1}{s+1}  {{m+1}\brack{t}}_q  {{m-t+1}\brack{1}}_q  {{m-t}\brack{1}}_q{ q \choose s} \\
&=  \frac{\left(q^{m+1} - 1\right)}{(s+1)(q-1)} {{m}\brack{t}}_q {\stirling{m-t}{1}}_q   {q \choose s}.
\end{align*}
By combining this with \eqref{eq:MphiM}, we obtain the desired formula. \hfill $\Box$
	
	\begin{rmk}\label{rmk:Nathan}
		{\rm
		The formula \eqref{NoMinWtPRM} for the number of minimum weight codewords of the projective Reed-Muller code $\prm_q(d,m)$ is evidently a positive integer when $s=0$. In case $s\ne 0$, we can easily rewrite \eqref{NoMinWtPRM}~as 
		$$
		\frac{\left(q^{m+1} - 1\right)(q^m - 1)}{(q+1)(q-1)}  {{m-1}\brack{t}}_q  {{q+1} \choose {s+1}},
		$$
		and this readily shows  that it is a positive integer. In the special case when  $t=m-1$ and $s= q-k$, where $2\le k \le q-1$, the formula reduces to
		$$
		\frac{(q-1)}{(q+1)} p^{ }_m p^{ }_{m-1} {{q+1} \choose {k}}.
		$$
		It may be interesting to note that this is precisely equal to $(q-1)$ times the  number of collections of $k$ collinear points in $\PP^m(\Fq)$. For geometric interpretations such as this, and more, we refer to  Kaplan and Matei \cite{KM}, especially \S 3.3 of their paper. 
	}
	\end{rmk}

\section{Dimension Formulas for Projective Reed-Muller Codes}

Fix throughout this section $m, d \in \Z$ such that $m\ge 1$ and $1\le d  \le m(q-1)+1$.
As before, let $\, t , s$ be unique integers such that 
$d-1=t(q-1)+s$ with $t\ge 0$ and $0\le s< q-1$.
An explicit formula for the dimension of a PRM code of an arbitrary order $d$ was first given by S\o rensen \cite{S}. 
It is as follows.
\begin{eqnarray}\label{eq:SdimPRM}
\displaystyle{\alpha^{ }_q(d,m) := \sum_{\substack{e=1 \\ e \equiv d\Mod{q-1} }}^d \left( \sum_{j=0}^{m+1} (-1)^j \binom{m+1}{j}\binom{e-jq+m}{e-jq} \right).}
\end{eqnarray}

Later Mercier-Rolland \cite{MR} gave another formula for the dimension of PRM codes of order $d$ using the dimension of vanishing ideal of $\PP^m(\Fq)$. This one is as follows. 
\begin{eqnarray}\label{eq:MRdimPRM}
\beta_q(d,m) :=  \binom{m+d}{d} - 
\sum_{j=2}^{m+1} (-1)^j {{m+1}\choose{j}} \sum_{i=0}^{j-2} {{d+(i+1)(q-1) - jq + m}\choose{d+(i+1)(q-1) - jq }}.
\end{eqnarray}
Around the same time\footnote{
A comment on the dates of publication of \cite{MR} and \cite{RT} can be found in  \cite[Remark~3.4]{G2}.},  Renter\'{i}a and Tapia-Recillas \cite[Prop. 12]{RT} 
gave the following nice formula for the dimension of $\prm_q(d,m)$:
\begin{equation}\label{eq:RTdimPRM}
\gamma^{ }_q(d,m) :=  \sum_{i=0}^m\sum_{j=0}^{i}(-1)^j \binom{i}{j}\binom{i+d-1 -jq}{i}.
\end{equation}
More recently, Can, Joshua and Ravindra \cite[Prop. 3.5]{CJR} have proposed another formula for the dimension of PRM codes of order $d$, which is as follows. 
\begin{eqnarray}\label{eq:CdimPRM}
\displaystyle{\delta^{ }_q(d,m) := \sum_{\substack{e=1 \\ e \equiv d\Mod{q-1} }}^d \left( \sum_{j=0}^{\lfloor 	e/q \rfloor } (-1)^j \binom{m+1}{j}\binom{e-jq+m}{m} \right).}
\end{eqnarray}
Can, Joshua and Ravindra \cite{CJR} further alleged that the formula of 
S\o rensen \cite{S} is wrong and wrote that the expression on the right in \eqref{eq:SdimPRM} is equal to $0$. This is, in fact, not correct. To clear the confusion, let us first note that the binomial coefficients are (and always should be!) defined as in the beginning of Section~\ref{sec2}. In particular, ${a\choose b}$ is zero when  $b<0$.  In fact, if $a, b\in \Z$, then
\begin{equation}
\label{binom1}
{a\choose b}=0 \Longleftrightarrow \text{either } b < 0 \text{ or } b> a\ge 0.
\end{equation}
Furthermore, the identity ${a\choose b}={a\choose {a-b}}$ is not absolute, but conditional. In the first place it requires that $a\in \Z$ and in fact, for any $a,b\in \Z$
\begin{equation}
\label{binom2}
{a\choose b}={a\choose {a-b}} \Longleftrightarrow \text{either } a \ge 0 \text{ or } a < b < 0.
\end{equation}
We remark in passing that on the other hand, the Pascal triangle identity, viz., 
\begin{equation}
\label{eq:PascalTriangle}
{a\choose {b-1}} + {a\choose {b}} = {{a+1}\choose {b}}
\end{equation}
is an absolute identity in the sense that it is valid for all $b\in \Z$ and all $a$ in any integral domain containing $\Z$ as a subring. At any rate, in view of \eqref{binom1}, the summation over $j$ in \eqref{eq:SdimPRM} 
is, in effect, over $0\le j \le \min\{m+1, \lfloor 	e/q \rfloor\} = \lfloor 	e/q \rfloor$ for 
any positive integer $e\le d$, where the last equality follows since 
$e\le d \le m(q-1)+1 \le mq$.  
 Moreover,  for $0\le j \le \lfloor 	e/q \rfloor$,  by \eqref{binom2}, we obtain  
 $\binom{e-jq+m}{e-jq}  = \binom{e-jq+m}{m}$. 
 Thus it is clear that $\alpha^{ }_q(d,m) =\delta^{ }_q(d,m)$. In other words, the formulas of S\o rensen and of Can, Joshua and Ravindra are, in fact, the same! The argument in Section~3 of \cite{CJR} to conclude that $\alpha^{ }_q(d,m) =0$ uses their Lemma 3.3, which is perfectly fine, but then it uses the ``identity"  $\binom{e-jq+m}{e-jq}  = \binom{e-jq+m}{m}$ for arbitrary nonnegative integer $j$, which as we have noted in \eqref{binom2}, is not always true. [There is also a typo in the definition given in \cite[p. 6]{CJR} of the polynomial $P(x)$; namely,  $t$ should be replaced throughout by $t+m$ in the expression on the right, but this is minor]. 
 
The equivalence of  \eqref{eq:SdimPRM} and \eqref{eq:MRdimPRM} is rather nontrivial. As mentioned in \cite{MR}, a direct proof of this was given by Michel Quercia, but to the best of our knowledge, this has not been published. The argument of Michel Quercia was made available to us by Robert Rolland, and we have reproduced it in the appendix to this paper. 
We give below an alternative argument to prove the equivalence of these formulas.  This alternative argument uses the following simple combinatorial lemma,  which has been given, for instance,  by S\o rensen \cite{S}. 

\begin{lemma}
\label{lem:SorComb}
Let $n,a,b$ be any nonnegative integers and let $N(a,n,b)$ denote the number of ways in which one can place $a$ objects in $n$ blocks such that no block contains more than $b$ objects.  Then
\begin{equation}\label{eq:Nanb}
N(a,n,b)=\sum_{j=0}^{n} (-1)^j {{n}\choose{j}} {{a-j(b+1) + n-1}\choose{a-j(b+1)}}.
\end{equation}
\end{lemma}

\begin{proof}
See \cite[Lem. 5]{S}.  
Or use an argument as in the proof of \cite[Lem. 1]{GR}.
\end{proof}

\begin{rmk}\label{rmk:Nrho}
	{\rm 
		Given any nonnegative integers $n, a,b$, it is clear from the description of $N(a,n,b)$ that it is the number of monomials in $n$ variables of degree $a$ such that the degree in each variable is at most $b$. In particular, if $0 \le \nu \le n(q-1)$, then  $N(\nu, n, q-1)$ is  the number of reduced monomials of degree $\nu$ in $n$ variables, and so we see that 
	$N(\nu, n, q-1) = \rho_q(\nu , n) - \rho_q(\nu -1, n)$, where  $\rho_q(\nu, n)$ is given by  \eqref{eq:dimRM}  if $n\ge 1$ and 
	where by convention, $\rho_q(0,0):=1$ and $\rho_q(k,n):=0$ if $k\in \Z$ with $k<0$.
This can also be seen directly by comparing  \eqref{eq:dimRM} with the formula in Lemma~\ref{lem:SorComb}               and using the Pascal triangle identity \eqref{eq:PascalTriangle}. 
Note also that the degree of a reduced monomial in $n$ variables is always $\le n(q-1)$ and~so  
\begin{equation}\label{eq:rho-with-min}
\rho_q(\nu , n) = \rho_q\left(\min\{\nu, \, n(q-1)\}, \; n\right) \quad \text{for any nonnegative integers } \nu, \, n.
\end{equation}
This is clear from the description of $\rho_q(\nu , n)$ as the number of reduced monomials in $n$ variables of degree $\le \nu$, and also from the formula \eqref{eq:dimRM} because if $\nu > n(q-1)$, then using \eqref{eq:dimRM}, \eqref{eq:Nanb},  and successive applications of 
\eqref{eq:PascalTriangle}, we see that 
$$
\rho_q(\nu , n) = \rho_q(n(q-1) , \, n) + \sum_{a= n(q-1) +1}^{\nu} N(a, \,  n, \, q-1) \; 
= \; \rho_q(n(q-1),\,  n),
$$
where the last equality follows since any distribution of $n(q-1)+1$ or more objects in $n$ blocks will result in at least one block containing $q$ or more objects. 
}
\end{rmk}
We are now ready to prove that the two dimension formulas $\alpha^{ }_q(d,m)$ and $\beta_q(d,m)$ given by \eqref{eq:SdimPRM}  and \eqref{eq:MRdimPRM} are equal. 
\begin{prop}
\label{pro:equality}
$\alpha^{ }_q(d,m) = \beta_q(d,m)$.
\end{prop}

\begin{proof}
First note that we may rewrite \eqref{eq:MRdimPRM} as 
\begin{eqnarray*}
\beta_q(d,m)= \binom{m+d}{d} - 
\sum_{j=0}^{m+1} (-1)^j {{m+1}\choose{j}} 
\sum_{\substack{e=d+(q-1) \\ e \equiv d\Mod{q-1} }}^{d+(j-1)(q-1)} {{e - jq + m}\choose{e- jq }}.
\end{eqnarray*}
Also note that by separating terms in \eqref{eq:SdimPRM} corresponding to $e=d$ and $j=0$ as well as 
$e=d$ and $j\ge 1$,  we can write 
\begin{eqnarray*}
\alpha^{ }_q(d,m) & =&
 \sum_{\substack{1\le  e < d \\ e \equiv d\Mod{q-1} }}
 \sum_{j=0}^{m+1} (-1)^j {{m+1}\choose{j}}  {{e - jq + m}\choose{e- jq }} \\
 & & \qquad  + \binom{m+d}{d}
 + \;  \sum_{j=1}^{m+1} (-1)^j {{m+1}\choose{j}}  {{d - jq + m}\choose{d- jq }}
\end{eqnarray*}
Thus $\alpha^{ }_q(d,m)=\beta_q(d,m)$ is equivalent to the equality 
\begin{eqnarray}\label{eq:equivdimPRM}
\sum_{j=0}^{m+1} (-1)^j {{m+1}\choose{j}} \sum_{\substack{e=1 \\ e \equiv d\Mod{q-1} }}^{d+(j-1)(q-1)} {{e - jq + m}\choose{e- jq }}=0, 
\end{eqnarray}
Define
\begin{eqnarray*}
c:=\sum_{j=0}^{m+1} (-1)^j {{m+1}\choose{j}} \sum_{\substack{e \equiv d\Mod{q-1} \\ -\infty<e\le d+(j-1)(q-1)}}{{e - jq + m}\choose{e- jq }}.
\end{eqnarray*}
Then the expression on the left in  \eqref{eq:equivdimPRM} is equal to $c-1$ if $d$ is a multiple of $q-1$, and is equal to $c$ otherwise. Thus it is enough to show that $c=1$ if $d$ is a multiple of $q-1$ and $c=0$ if $d$ is not a multiple of $q-1$. 
Now
\begin{eqnarray*}
c & = &  \sum_{j=0}^{m+1} \sum_{k=0}^\infty (-1)^j {{m+1}\choose{j}} {{d+(j-1)(q-1)-k(q-1) - jq + m}\choose{d+(j-1)(q-1)-k(q-1)-jq }}\\
& = & \sum_{k=0}^\infty \sum_{j=0}^{m+1} (-1)^j {{m+1}\choose{j}} {{d+(j-1)(q-1)-k(q-1) - jq + m}\choose{d+(j-1)(q-1)-k(q-1)-jq }}\\
& = & \sum_{k=0}^\infty \sum_{j=0}^{m+1} (-1)^j {{m+1}\choose{j}} {{d-(k+1)(q-1)-j + m}\choose{d-(k+1)(q-1)-j }}\\
& = & \sum_{k=0}^\infty N(d-(k+1)(q-1),\; m+1,\; 0)
\end{eqnarray*}
where the last equality uses Lemma~\ref{lem:SorComb} and the notation therein.  But evidently, 
\begin{eqnarray*}
N(d-(k+1)(q-1),\; m+1,\; 0)=\begin{cases}
1 & \text{if } d=(k+1)(q-1), \\
0 & \text{if } d\neq (k+1)(q-1).
\end{cases}
\end{eqnarray*}
This implies that $c=1$ if $d$ is a multiple of $q-1$,  and $0$ otherwise. 
\end{proof}

We will now consider the formula $\gamma_q(d,m)$ of Renter\'{i}a and Tapia-Recillas  \cite[Prop.~12]{RT}.  They proved it 
 using an ideal-theoretic approach and the notion of Hilbert functions.
 (See also \cite[Cor.~3.9 (b)]{CNL}, where Carvalho, Neumann, and Lopez derive the same formula  from a more general result on projective nested cartesian codes.)
 Moreover, \cite[Lem.~9]{RT} implies that $\gamma_q(d,m)$ can be related to the dimensions of certain Reed-Muller codes. We show below that the notion of projective reduction can also be used to give a short direct proof of the formula  of Renter\'{i}a and Tapia-Recillas for the dimension of  $\prm_q(d,m)$. 

Recall that we had defined the Reed-Muller code $\RM_q(\nu, m)$ for any $\nu\in \Z$ with $0\le \nu \le m(q-1)$. We shall find it convenient to extend this definition a bit so as to set 
$\RM_q(0,0)$ to be the $1$-dimensional code of length $1$, namely, $\Fq$.
This is  consistent with the formula \eqref{eq:dimRM}  
given earlier, namely, 
$$
\rho_q(\nu, m) := \dim \RM_q(\nu,m)
=\sum_{j=0}^{m}(-1)^j \binom{m}{j}\binom{m+\nu -jq}{m}
$$
and also with convention $\rho_q(0,0):= 1$, which was stated earlier. 

\begin{thm}\label{newdimformula}
The dimension of the code $\prm_q(d,m)$ is given by
\begin{equation}
\label{eq:deltaq}
\sum_{i=0}^m \rho_q(d-1,\, i) = \sum_{i=0}^m\sum_{j=0}^{i}(-1)^j \binom{i}{j}\binom{i+d-1 -jq}{i}   = \gamma_q(d,m).
\end{equation}
\end{thm}
\begin{proof}
Let $E$ be 
the set of all $(m+1)$-tuples of nonnegative integers whose sum is equal to $d$. 
For $\mathbf{a} = (a_0,\ldots,a_m)\in E$, let 
$\ell(\mathbf{a} ): = \max\{j \in\{0,1, \dots , m\} : a_j > 0\}$.
Note that this is well-defined since $d\ge 1$. 
Define 
$$
B:=\{\mathbf{a} =(a_0,\ldots,a_m)\in E:a_j<q \text{ for }0\le j <\ell (\mathbf{a}) \}. 
$$
It is clear from  Lemma \ref{lemp3} that $\dim \prm_q(d,m)=|B|$. Now observe that
$$
B= \coprod_{i=0}^m B_i \quad \text{where for } 0\le i \le m, \quad B_i:=\{ \mathbf{a} \in B: \ell (\mathbf{a} ) =i \}.
$$
For  $i\in \{0,1,\dots ,m\}$, the map $\mathbf{a} \mapsto (a_0, \dots, a_{i-1})$ is readily seen to be a bijection of $B_i$ onto the set of exponent vectors of reduced monomials in $i$ variables 
	of degree $\le d-1$; consequently, $|B_i|=\rho_q(d-1,\, i) $.
	This yields the desired formula. 
\end{proof}

Finally,  we give a direct proof of the equivalence of the dimension formulas of S\o rensen and
Renter\'{i}a--Tapia-Recillas.

\begin{prop}
	\label{pro:equalitywithnew}
	$\alpha^{ }_q(d,m) = \gamma_q(d,m)$.
\end{prop}

\begin{proof}
	First, note that \eqref{eq:SdimPRM} can be written in the notation of Lemma~\ref{lem:SorComb} as 
\begin{equation}
\label{eq:alphaN}
	\alpha^{ }_q(d,m) = \! \sum_{\substack{e=1 \\ e \equiv d\Mod{q-1} }}^d \! N(e, m+1, q-1). 
\end{equation}
Fix any $e\in \{1, \dots , d\}$. Then $N(e, m+1, q-1)$ is the number of reduced monomials of degree $e$ in $m+1$ variables
$x_0, \dots , x_m$. 
Given such a reduced monomial, if $x_i$ is the last variable appearing in it, then it must be of the form $\mu\, x_i^{j_i}$, where 
$1\le j_i \le q-1$ and $\mu$ is a reduced monomial in $x_0, \dots , x_{i-1}$ of degree $e-j_i$. Note that the degree $e-j_i$ of $\mu$ 
satisfies $e-q+1 \le e-j_i \le e-1$. Conversely, any reduced monomial in $x_0, \dots , x_{i-1}$ of degree $a$, where $e-q+1 \le a \le e-1$, 
would give rise (upon multiplication by $x_i^{e-a}$) to a unique reduced monomial of degree $e$ in $x_0, \dots , x_{i}$ with a positive exponent for $x_i$. Thus partitioning reduced monomials of degree $e$ in $x_0, \dots , x_m$ by the last variable appearing in them, we see that 
$$
N(e, m+1, q-1) = \sum_{i=0}^m \big( \rho_q(e-1, \, i)  -  \rho_q(e-q,\, i) \big).
$$
On the other hand, since $d = t(q-1)+s+1$ with $0\le t \le m$ and $0\le s < q-1$, the parameter $e$ in the summation in  \eqref{eq:alphaN} takes values of the form $r(q-1)+s+1$, where $r$ ranges from $0$ to $t$. Thus, 
we see that  
$$
	\alpha^{ }_q(d,m) = \sum_{r=0}^t N(r(q-1)+s+1, \, m+1, \, q-1)  = \sum_{r=0}^t \sum_{i=0}^m \big( A(r,i) - A(r-1, i)\big),
	$$
where $A(k,i):= \rho_q(k(q-1)+s, \, i)$ for $-1\le k \le t$ and $0\le i \le m$. 	Now, by interchanging the summations over $i$ and $r$, and then observing that the inner sum is telescoping, and moreover, 
$A(-1, i) = 0$, we obtain
$$
\alpha^{ }_q(d,m) =\sum_{i=0}^m A(t,i) = \sum_{i=0}^m \rho_q(d-1, \, i).  
$$
This proves that $\alpha^{ }_q(d,m) = \gamma_q(d,m)$.
\end{proof}

\section*{Acknowledgements}
We are grateful to Robert Rolland and  Michel Quercia for sharing with us a direct proof of 
the equivalence of formulas \eqref{eq:SdimPRM} and \eqref{eq:MRdimPRM}, and we thank  M.  Quercia for permitting us to reproduce his proof in the appendix below. We also thank Nathan Kaplan for his interest in this work and for bringing  \cite{KM}  and the geometric interpretation mentioned in Remark~\ref{rmk:Nathan} to our attention. Thanks are also due to Hiram Lopez and the reviewers of a preliminary version of this article for some useful remarks and pointers to relevant literature. 

\section*{Appendix: Quercia's Proof of Equivalence of Dimension Formulas}

Here we reproduce Michel Quercia's original proof of the equivalence of the formulas,  due to 
S\o rensen and Mercier-Rolland,  for the dimension of projective Reed-Muller codes. We use notations consistent with this article.

\begin{prop}[Quercia]
Let $q$ be a prime power and let $d,m$ be positive integers such that $d\le m(q-1) + 1$.  
Also, let 
 $\alpha^{ }_q(d,m)$ and $\beta_q(d,m)$ be given by \eqref{eq:SdimPRM} and \eqref{eq:MRdimPRM}, respectively.  Then $\alpha^{ }_q(d,m) = \beta_q(d,m)$. 
\end{prop}

\begin{proof}
As in the proof of Proposition \ref{pro:equality},  the assertion 
$\alpha^{ }_q(d,m)=\beta_q(d,m)$ is equivalent to the equality 
\begin{eqnarray}\label{eq:QequivdimPRM}
\sum_{j=0}^{m+1} (-1)^j {{m+1}\choose{j}} \sum_{\substack{e \equiv d\Mod{q-1} \\
0<e \le d+(j-1)(q-1)}} {{e - jq + m}\choose{e- jq }}=0.
\end{eqnarray}
Define
\begin{eqnarray}\label{eq:defc}
c:=\sum_{j=0}^{m+1} (-1)^j {{m+1}\choose{j}} \sum_{\substack{e \equiv d\Mod{q-1} \\ -\infty<e\le d+(j-1)(q-1)}}{{e- jq + m}\choose{e- jq }}.
\end{eqnarray}
Observe that the expression on the left in  \eqref{eq:QequivdimPRM} is equal to $c-1$ if $d$ is a multiple of $q-1$, and is equal to $c$ otherwise. Thus it is enough to show that $c=1$ if $d$ is a multiple of $q-1$ and $c=0$ if $d$ is not a multiple of $q-1$. We use 
a functional approach to prove this. 

Let $\mathcal{E}$ be the $\Z$-module of the sequences $u=(u_n)_{n\in \Z}$ which take integer values and vanish in the vicinity of $-\infty$ (i.e., for each $u=(u_n)_{n\in \Z} \in \mathcal{E}$,  there is $n_0\in \Z$ such that 
$u_n =0$ for all $n \in \Z$ with $n< n_0$). Denote by $I$ the identity mapping of $\mathcal{E}$ and $T$ the transfer operator defined by
$$
(Tu)_n=u_{n+1} \quad \text{for }n\in\Z.
$$
The sequence $u^{(m)}$, defined by
$$
(u^{(m)})_n=\binom{n+m}{n} \quad \text{for }n\in\Z,
$$
belongs to $\mathcal{E}$ and the number $c$ as defined in \eqref{eq:defc} is the term with index $0$ of the sequence $x=Lu^{(m)}$ where 
$$
L=\sum_{j=0}^{m+1} (-1)^j {{m+1}\choose{j}} \sum_{\substack{e \equiv d\Mod{q-1} \\ -\infty<e \le d+(j-1)(q-1)}} T^{e-jq}.
$$ 
Note that the above series has a value as an $\mathcal{E}$-endomorphism. Define 
$$
\Delta:=\sum_{k\le 0} T^{k(q-1)}=(I-T^{-q+1})^{-1}.
$$ 
Then
\begin{eqnarray*}
&L&=\sum_{j=0}^{m+1} (-1)^j {{m+1}\choose{j}} T^{d+(j-1)(q-1)-jq}\Delta\\
&&=\sum_{j=0}^{m+1} (-1)^j {{m+1}\choose{j}} T^{d-j-q+1}\Delta\\
&& =T^{d-q+1}\Delta (I-T^{-1})^{m+1}
\end{eqnarray*}
On the other hand, $(I-T^{-1})u^{(m)}=u^{(m-1)}$. This implies that
\begin{eqnarray*}
(I-T^{-1})^{m+1}u^{(m)}=u^{(-1)},
\end{eqnarray*}
where $u^{(-1)}$ is the sequence belonging to $\mathcal{E}$ defined by
$$
(u^{(-1)})_0=1 \quad \text{and} \quad (u^{(-1)})_n=0 \text{ if $n\in \Z$ with }n\ne 0.
$$
Thus $x=Lu^{(m)}=T^{d-q+1}\Delta u^{(-1)}$ is given by
$$
x_n=\begin{cases}
1 & \text{if }n \equiv -d (\text{mod }q-1) \text{ and } n > -d, \\
0 & \text{otherwise}.
\end{cases}
$$
Considering $x_0$, we obtain the desired result.
%
\end{proof}


\begin{thebibliography}{AAAA}
\addtocontents{toc}{\hfill{}\\{\bf References} \hfill{\bf \thepage}\\}
	

\bibitem{AssKey}
E. F. Assmus Jr. and J. D. Key, \emph{Designs and their Codes}, Cambridge Tracts in Math., {\bfseries 103},  Cambridge Univ. 
Press, 1992.


\bibitem{BR}
S.  Ballet and R.  Rolland,  On low weight codewords of generalized affine and projective Reed-Muller codes,  \emph{Des. Codes Cryptogr.} {\bfseries 73} (2014), 271--297.

\bibitem{BD}
P. Beelen and M. Datta, Generalized Hamming weights of affine Cartesian codes, \emph{Finite Fields Appl.} {\bfseries 51} (2018),
130--145.

\bibitem{BDG1}
P. Beelen, M. Datta and S. R. Ghorpade, Vanishing ideals of projective spaces over finite fields and a projective footprint bound, \emph{Acta Math. Sin. (Engl. Ser.)} {\bfseries 35} (2019), 
47--63. 

\bibitem{BDG2}
P. Beelen, M. Datta and S. R. Ghorpade,   A combinatorial approach to the number of solutions of systems of homogeneous polynomial  equations  over finite fields, \emph{Moscow Math J.} {\bfseries 22} (2022), 565--593.

\bibitem{B}
T. Berger, 
Automorphism groups of homogeneous and projective Reed-Muller codes, \emph{IEEE Trans. Inform. Theory} {\bfseries  48} (2002), 
1035--1045.

\bibitem{BC}
T. Berger and P. Charpin, The automorphism group of generalized Reed-Muller
codes, \emph{Discrete Math.} {\bfseries 117} (1993), 1--17.

%

	
	
%
		

\bibitem{CJR}
M. B. Can, R. Joshua and G. V. Ravindra, Higher Grassmann codes II, 
\emph{Finite Fields Appl.} {\bfseries 89} (2023), Art. 102211, 21 pp.


\bibitem{CNL}
C.  Carvalho,  V. G. L. Neumann and H. H. Lopez,  Projective nested cartesian codes, 
\emph{Bull. Braz. Math. Soc.} \textbf{48} (2017), 283--302. 

\bibitem{DGM}
P. Delsarte, J. M. Goethals and F. J. MacWilliams, On generalized Reed-Muller codes and their relatives, \emph{Information and Control} {\bfseries 16} (1970), 403--442.



\bibitem{G2}
S. R. Ghorpade, A note on Nullstellensatz over finite fields, in: \emph{Contributions in Algebra and Algebraic Geometry} (Aurangabad, 2017), \emph{Contemp. Math.}, {\bf 738}, Amer. Math. Soc., Providence, 2019, pp. 23--32. 

\bibitem{GR}
S. R. Ghorpade and R. Ludhani, On the purity of  resolutions of Stanley-Reisner rings associated to Reed-Muller codes,  in: \emph{Algebra and Related Topics with Applications}  (Aligarh,  2019), 
\emph{Springer Proc. Math.  Stat.} {\bf 392}, 
Springer, Singapore, 2022, pp. 325--335.

\bibitem{GRRT}
S. R. Ghorpade, C. Ritzenthaler, F. Rodier and M. A. Tsfasman,  
Arithmetic, geometry, and coding theory: Homage to Gilles Lachaud, \emph{Contemp. Math.}, {\bfseries 770}, Amer. Math. Soc., Providence, 2021, 131--150.

%
%
%

\bibitem{HP} P.~Heijnen and R. Pellikaan,  {Generalized Hamming weights of $q$-ary Reed-Muller codes}, \emph{IEEE Trans. Inform. Theory} {\bf 44} (1998), 181--196.
%
%
%
	



\bibitem{KM}
N. Kaplan and V. Matei, Counting plane cubic curves over finite fields with a prescribed number of rational intersection points, 
\emph{Eur. J. Math.} {\bfseries 7} (2021), 
1137--1181.

\bibitem{KLP}
T. Kasami, S. Lin and W. W. Peterson, 
New generalization of the Reed-Muller codes--Part I: Primitive Codes,  \emph{IEEE Trans. Inform. Theory} {\bfseries IT-14} (1968), 189--199.

\bibitem{L}
G. Lachaud,
Projective Reed-Muller codes, in: \emph{Coding Theory and Applications} (Cachan, 1986), Lecture Notes in Comput. Sci., {\bfseries 311}, Springer, Berlin, 1988, pp. 125--129.

\bibitem{L2}
G. Lachaud, The parameters of projective Reed-Muller codes,
\emph{Discrete Math.}  {\bf 81} (1990), 
217--221.

\bibitem{Le} 
E. Leducq,
A new proof of Delsarte, Goethals and MacWilliams theorem on minimal weight codewords of generalized 
Reed-Muller codes,  \emph{Finite Fields Appl.} {\bfseries 18} (2012), 581--586.



\bibitem{MR}
D. J. Mercier and R. Rolland, Polyn\^{o}mes homog\`{e}nes qui s$'$annulent sur l$'$espace projectif $\PP^m(\Fq)$. \emph{J. Pure Appl. Algebra} {\bfseries 124} (1998), 227--240.

\bibitem{M}
D. E. Muller,  Application of Boolean algebra to switching circuit design and to error detection, \emph{IRE Trans. Electron. Comput.}  
{\bfseries{EC-3}} (1954), 6--12.

\bibitem{R}
I. S. Reed, A class of multiple-error-correcting codes and the decoding scheme,  
\emph{IRE Trans. Inform. Theory}
{\bfseries 4} (1954), 38--49. 

\bibitem{R2}
I. S. Reed, A brief history of the development of error correcting codes, \emph{Comput. Math. Appl.} {\bfseries 39} (2000), 
89--93.

\bibitem{RT}
C. Renter\'{i}a and H. Tapia-Recillas, Reed-muller codes: an ideal theory approach, \emph{Comm. Algebra} \textbf{25}  (1997), 
401--413. 

\bibitem{RR}
R. Rolland, Number of points of non-absolutely irreducible hypersurfaces, in: 
\emph{Algebraic Geometry and its Applications} (Papette, 2007),  \emph{Ser. Number Theory Appl.} {\bfseries 5},
World Scientific, Singapore, 2008, pp. 481--487.

\bibitem{Se}
J.-P. Serre, Lettre \'{a} M. Tsfasman, \emph{Journ\'{e}es Arithm\'{e}tiques (Luminy,1989)}, Ast\'{e}risque No. {\bfseries 198-199-200}
(1991), 351--353.

\bibitem{S}
A. B. S{\o}rensen, Projective Reed-Muller codes, \emph{IEEE Trans. Inform. Theory} {\bfseries 37} (1991), 1567--1576.

\bibitem{VM} 
S. G. Vl\`eduts and Yu.  I.  Manin, Linear codes and modular curves, \emph{J. Sov. Math.}    \textbf{30} (1985), 2611--2643.



	





	
	



	
	
	
%
		







	



%
%

	






	
\end{thebibliography}
\end{document}